\documentclass[10pt,journal,compsoc,onecolumn]{IEEEtran}

\ifCLASSOPTIONcompsoc
  \usepackage[nocompress]{cite}
\else
  \usepackage{cite}
\fi
\usepackage{graphics, graphicx}
\usepackage{wrapfig,lipsum}

\usepackage{amsmath, amsthm, amssymb}
\usepackage{amssymb}
\usepackage{amsmath}
\usepackage{textcomp}

\usepackage{tikz}       
\usetikzlibrary{arrows,decorations.pathmorphing,backgrounds,positioning,fit,petri}
\usetikzlibrary{graphs}
\usetikzlibrary{trees,snakes}
\usepackage{xcolor}
\definecolor{raspberry}{rgb}{0.89, 0.04, 0.36}
\definecolor{deeppink}{rgb}{1.0, 0.08, 0.58}
\definecolor{darkgreen}{rgb}{0.13, 0.55, 0.13}
\newcommand{\leo}{}
\newcommand{\rev}{}

\DeclareMathAlphabet{\mathcalligra}{T1}{calligra}{m}{n}

\newtheorem{definition}{Definition}
\newtheorem{theorem}{Theorem}[section]

\newtheorem{proposition}[theorem]{Proposition}
\newtheorem{corollary}[theorem]{Corollary}
\newtheorem{observation}[theorem]{Observation}
\usepackage[margin=1in]{geometry}
\usepackage{float}
\usepackage{subfigure}
\usepackage{amsfonts,mathrsfs,anysize,fancyhdr,epsfig}

\ifCLASSINFOpdf
\else
\fi
\begin{document}
%
\title{Nonbinary tree-based\\ phylogenetic networks}
%
%
%
%

\author{Laura Jetten and
        Leo van Iersel
\IEEEcompsocitemizethanks{\IEEEcompsocthanksitem E-mail: ljetten@outlook.com and l.j.j.v.iersel@gmail.com
\IEEEcompsocthanksitem Delft University of Technology, PO-box 5,
2600 AA, Delft,\protect\\ The Netherlands.\IEEEcompsocthanksitem Leo van Iersel was partly supported by a Vidi grant from the Netherlands Organisation for Scientific Research (NWO) and partly by the 4TU Applied Mathematics Institute (4TU.AMI).}}
\IEEEtitleabstractindextext{%
\begin{abstract}
Rooted phylogenetic networks are used to describe evolutionary histories that contain non-treelike evolutionary events such as hybridization and horizontal gene transfer. In some cases, such histories can be described by a phylogenetic base-tree with additional linking arcs, which can for example represent gene transfer events. Such phylogenetic networks are called \emph{tree-based}. Here, we consider two possible generalizations of this concept to nonbinary networks, which we call \emph{tree-based} and \emph{strictly-tree-based} nonbinary phylogenetic networks. We give simple graph-theoretic characterizations of tree-based and strictly-tree-based nonbinary phylogenetic networks. Moreover, we show for each of these two classes that it can be decided in polynomial time whether a given network is contained in the class. Our approach also provides a new view on tree-based binary phylogenetic networks. \rev{Finally, we discuss two examples of nonbinary phylogenetic networks in biology and show how our results can be applied to them.}
\end{abstract}

}

\maketitle

\IEEEdisplaynontitleabstractindextext

%
\IEEEpeerreviewmaketitle

\IEEEraisesectionheading{\section{Introduction}\label{sec:introduction}}

Rooted phylogenetic networks are becoming increasingly popular as a way to describe evolutionary histories that cannot be described by a phylogenetic tree~\cite{book,expanding}. The leaves of such a network are labelled and represent, for example, currently-living species, while the root of the network represents a common ancestor of those species. Vertices with two or more outgoing arcs represent a divergence event in which a lineage split into two or more lineages, while vertices with two or more incoming arcs represent a convergence of different lineages into a single lineage. The latter events are called \emph{reticulate} evolutionary events and include, for example, hybridization, introgression and horizontal gene transfer. Therefore, these vertices are called \emph{reticulations}. A phylogenetic network without reticulations is a \emph{(rooted) phylogenetic tree}. Hence, phylogenetic networks are a more general model for evolutionary histories than phylogenetic trees.

Although the occurence of reticulate evolutionary events is well-accepted, there are different views on their importance. One possibility is to see evolution as a mainly tree-like (vertical) process with sporadic horizontal events. The other extreme is to completely \leo{abandon} the idea of a tree-of-life and to see evolution purely as a network~\cite{dagan,martin,doolittle,networkthinking}. This discussion is especially relevant for prokaryotes, where the main form of non-treelike evolution is horizontal gene transer, i.e. genetic material is transferred from one species to another coexisting species that is not a descendant. If the evolutionary history of a group of prokaryotes is mainly tree-like, then you could describe such a history as a phylogenetic species tree with additional cross-connecting arcs describing the horizontal gene transfer events. However, if their evolution is inherently network-like, then it might not be possible to identify any tree-like signal at all.

This discussion has \leo{recently} led to the introduction of a new class of phylogenetic networks called ``tree-based''~\cite{Artikel2}, which contains those networks that can be described by a phylogenetic base-tree with additional linking arcs between branches of the base-tree. This notion was motivated by the observation that this is not always possible, i.e. there exist networks that can not be described as a base-tree with linking arcs~\cite{blog}.

\leo{Francis and Steel showed recently} that there is a polynomial-time algorithm to decide whether a given binary phylogenetic network is tree-based or not~\cite{Artikel2}. In addition, it was shown that any phylogenetic network can be made tree-based by the addition of leaves. Hence, this notion has to be used with caution in the presence of possible extinctions or under-sampling. Even more recently, a simple graph-theoretic characterization was given that can also be used to decide whether a given binary network is tree-based or not~\cite{zhang}. Unfortunately, these results are all restricted to \emph{binary} phylogenetic networks, in which all vertices have at most two incoming and at most two outgoing arcs (see the next section for precise definitions). Moreover, the techniques used by these authors do not (easily) extend to nonbinary networks. 

Here, we also consider nonbinary phylogenetic networks. In such a network, a vertex can have more than two outgoing arcs, representing uncertainty in the order of divergence events, or more than two incoming arcs, representing uncertainty in the order of reticulate events. \rev{See Table~\ref{table} for real biological examples of nonbinary phylogenetic networks (also see~\cite{Violets,Hordeum,EukaryoticOrigin})}. In general, such uncertainties cannot simply be overcome by collecting more data~\cite{Bushes,Resolving}. Since the tree-basedness of nonbinary networks has not been introduced or studied before, we discuss different possible definitions of \leo{tree-based} in the nonbinary case.

\begin{table}
\rev{\begin{tabular}{|l||c|c|l|c|c|}
\hline
 & Reticulation number & Leaves & Reticulate Process & Tree-based & Reference\\ \hline \hline
Violets &21 & 16 & polyploidisation & no &\cite[Figure 4]{Violets2}, Figure~\ref{fig:violets} \\ \hline
Origin of Eukaryotes &6 &41& endosymbiosis & yes & \cite[Figure 2]{martin1999mosaic}, Figure~\ref{fig:eukaryotes}\\ \hline 
Influenza & 5 & 7 & reassortment & yes & \cite[Figure 1]{Influenza}\\ \hline 
Cichlids (fish)&5&19& hybridisation & yes & \cite[Figure 4]{Cichlids}\\ \hline
\end{tabular}}\smallskip
\caption{\label{table}\rev{Examples of nonbinary phylogenetic networks in biology. The \emph{reticulation number} is defined as the total number of ``additional branches'' in the network, i.e. a reticulation with~$p$ parents adds~$p-1$ to the reticulation number.}}
\end{table}

Roughly speaking, we call a nonbinary phylogenetic network \emph{strictly tree-based} if it can be obtained from a rooted (nonbinary) phylogenetic tree by adding linking arcs between the branches of the tree, such that no two linking arcs attach at the same point. Consequently, in such a network all vertices have at most two incoming arcs. In addition, all vertices with more than two outgoing arcs correspond to vertices of the base-tree, because the new vertices that are created by the addition of linking arcs all get at most two outgoing arcs: one of the base-tree and one linking arc. Hence, a strictly-tree-based network can be nonbinary only because the base-tree can be nonbinary.

We also consider \emph{tree-based} nonbinary phylogenetic networks, which are networks that can be obtained from a rooted (nonbinary) phylogenetic tree by adding linking arcs between branches and/or vertices of the tree. It turns out that \leo{a network is in this class precisely if it has at least one} binary refinement that is a tree-based binary phylogenetic network. This is a more general class than the strictly-tree-based variant.

Our main results are as follows. We first present an alternative view on binary tree-based phylogenetic networks, which can partly be extended to nonbinary networks. We \rev{introduce \emph{omnians}, which we define as non-leaf vertices of which all children are reticulations. We then use this notion to obtain a new, simpler characterization of binary tree-based phylogenetic networks.} We show that a binary phylogenetic network is \emph{tree-based} if and only if every subset~$S$ of its omnians has at least~$|S|$ different children. We use this to derive, in an alternative (independently discovered) way, the characterization of binary tree-based networks in terms of zig-zag paths~\cite{zhang} and a new matching-based algorithm for deciding whether a given binary network is tree-based. We also show that every binary network with at most two reticulations is tree-based and give a new sufficient condition for a binary network to be tree-based.

We then proceed to nonbinary networks. We show that our characterization of binary tree-based phylogenetic networks in terms of omnians can easily be generalized to the nonbinary case. We then obtain the first polynomial-time algorithm for deciding \leo{whether} a nonbinary phylogenetic network is tree-based. \leo{Additionally, we} show a simple counter example, showing that the characterization based on zig-zag paths can not be used to characterize nonbinary tree-based networks. However, we also show that nonbinary strictly-tree-based phylogenetic networks \emph{can} be characterized using zig-zag paths. Consequently, also for this class of networks it can be decided in polynomial time whether a given network belongs to the class.

We also discuss ``stable'' phylognetic networks~\cite{Artikel3}, in which for each reticulation~$r$ there exists some leaf~$x$ such that all paths from the root to~$x$ go through~$r$. We show that, although all binary stable phylogenetic networks are tree-based, this is not always the case for nonbinary networks.

This paper is organized as follows. We first give the definitions and new results for binary networks in Section~\ref{sec:bin}, then the definitions and results for nonbinary networks in Section~\ref{sec:nonbin}. \rev{Examples} of how these results can be applied to real, biological, nonbinary phylogenetic \rev{networks are} given in Section~\ref{sec:example}. We end with a discussion in Section~\ref{sec:discussion}.


\section{Binary phylogenetic networks}\label{sec:bin}
\subsection{Preliminaries}
\label{preliminaries}
First, some essential concepts around binary phylogenetic networks will be explained. Phylogenetic networks contain vertices and directed edges. Directed edges will be called arcs from now on. \\

\begin{definition}
A \emph{(rooted) binary phylogenetic network} is a directed acyclic graph $N=(V,A)$, which \leo{contains a single \emph{root} with indegree~0 and outdegree~1 or~2 and may in addition} contain the following types of vertices:
\begin{itemize}
\item vertices with outdegree 0, called \emph{leaves}, which are labelled;
\item vertices with indegree 2 and outdegree 1, called \emph{reticulations};
\item vertices with indegree 1 and outdegree 2, called \emph{tree-vertices}.
\end{itemize}
\end{definition}

An example of a binary phylogenetic network is given in Figure \ref{Examplenetwork}, \leo{in which leaves are coloured blue and reticulations are indicated with a pink shading around the nodes}.  A (\emph{rooted}) \emph{binary phylogenetic tree} is a binary phylogenetic network that contains no reticulations. 
Although every arc is drawn without arrow head, they are all directed to the lowest vertex. This is the case throughout the paper, unless explicitly mentioned otherwise.

\begin{figure}[h]\centering
\centering
\begin{tikzpicture}
\draw(0.55,5.2) node {Root};
\fill (0,5) circle (2pt); 
\fill[color=deeppink]  (-0.75,4)  circle (3.5pt);
\fill (-0.75,4) circle (2pt); 
\fill (0.375,4.5) circle (2pt); 
\fill (1.125,3.5) circle (2pt);
\filldraw ([xshift=-2pt,yshift=-2pt]-1.125,3.5) rectangle ++(4pt,4pt);
\fill (1.5,3) circle (2pt); 
\fill[color=deeppink] (-1.875,2.5) circle (3.5pt);
\fill (-1.875,2.5) circle (2pt); 
\fill [color=blue] (2.25,2) circle (2pt); 
\fill[color=deeppink] (-0.035,3.1) circle (3.5pt);
\fill (-0.035,3.1) circle (2pt);
\fill[color=deeppink] (-0.5,2) circle (3.5pt);
\fill (-0.5,2) circle (2pt);
\fill [color=blue] (-0.5,1.3) circle (2pt);
\filldraw ([xshift=-2pt,yshift=-2pt]-1.1,2.75) rectangle ++(4pt,4pt);
\fill [color=blue] (-2.4375,1.75) circle (2pt);
\draw [thick] (0,5)  -- (0.75,4) ; 
\draw [thick] (0,5)  -- (-0.75,4) ; 
\draw [thick] (-0.75,4)  -- (-1.5,3) ; 
\draw [thick] (0.75,4)  -- (1.5,3) ; 
\draw [thick] (1.5,3)  -- (2.25,2) ; 
\draw [thick] (-1.5,3)  -- (-2.4375,1.75) ; 
\draw [thick] (0.375,4.5) -- (-0.75,4);
\draw [thick] (1.125,3.5) -- (-0.035,3.1);
\draw [thick] (-0.035,3.1) -- (-1.875,2.5);
\draw [thick] (-1.125,3.5) -- (-0.03,3.1);
\draw [thick] (1.5,3) -- (-0.5,2);
\draw [thick] (-0.5,2) -- (-0.5,1.3);
\draw [thick]  (-1.1,2.75) -- (-0.5,2);
\draw (-2.43,1.45) node {$a$};      
\draw (-0.45,1) node {$b$};      
\draw (2.25,1.7) node {$c$};      
\draw (-1.2,2.95) node {$x$};      
\draw (-1.35,3.5) node {$y$};      

\end{tikzpicture}
\caption{An example of a binary phylogenetic network \leo{with leaf labels~$a$, $b$ and $c$, which can e.g.  represent three present-day species, and omnians~$x$ and~$y$}.}
\label{Examplenetwork}
\end{figure}
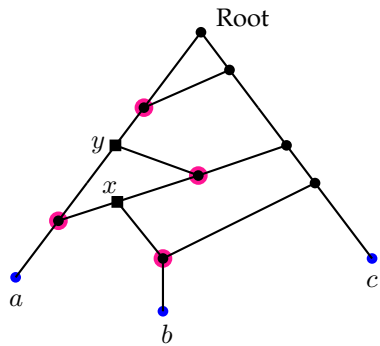

Take $(u,v) = a \in A$, an arc from vertex $u$ to $v$. Then, $a$ is called an \emph{outgoing} arc of~$u$~and an \emph{incoming} arc of~$v$.~Vertex~$u$~is a \emph{parent} of~$v$ and~$v$ is called a \emph{child} of~$u$. 
If there is also an arc~$(u,w) \in A$, then vertex~$w$~and~$v$ have a joint parent, so~$w$~and~$v$ are called \emph{siblings}.
When a non-leaf vertex $z$ has only reticulations as children, then $z$ is called an \emph{omnian}. For example in Figure \ref{Examplenetwork}, vertices $x$ and $y$ are omnians, since both children of these vertices are reticulations. \leo{Vertices can be omnian and reticulation at the same time, see e.g. vertices~$u$ and~$v$ in Figure~\ref{ExamplenetworknotTB}. Because of the importance of omnians, which will become clear later on in the paper, we always use square nodes for omnians and circular nodes for all other vertices.}

\begin{definition}
A binary phylogenetic network $N$ is \emph{tree-based} with base-tree $T$, when $N$ can be obtained from $T$ via the following steps:
\begin{enumerate}
\item[(i)] Add vertices to the arcs of $T$. These vertices, called \emph{attachment points}, have in- and outdegree 1.
\item[(ii)] Add arcs, called \emph{linking arcs}, between pairs of attachments points, so that~$N$ remains binary and acyclic.
\item[(iii)] Suppress every attachment point that is not incident to a linking arc.
\end{enumerate}
\end{definition}

\leo{Note in particular that it is not allowed to create multiple linking arcs between the same pair of attachment points since~$N$ is required to be a binary phylogenetic network.}

A binary phylogenetic network is \emph{tree-based} if it is tree-based with base-tree~$T$ for some binary phylogenetic tree~$T$. An example of the procedure is displayed in Figure \ref{definitionTB}. \rev{An example of a binary network that is not tree-based is given in Figure \ref{ExamplenetworknotTB}.
}

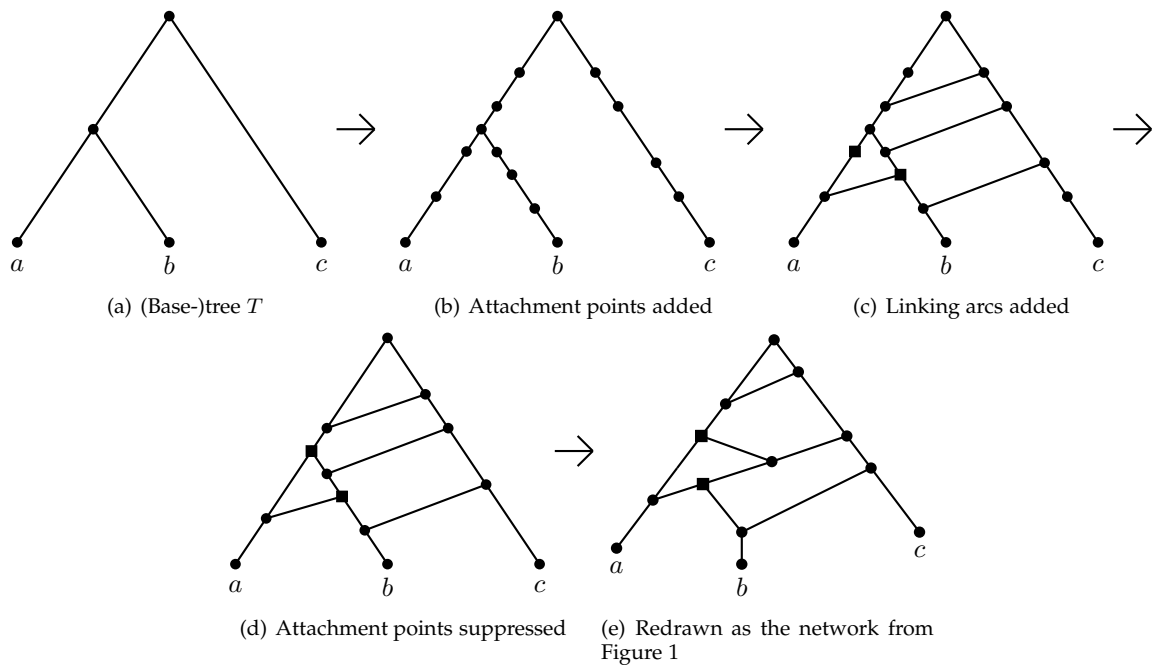
\begin{figure}[h]\centering
\subfigure[(Base-)tree $T$]{
\begin{tikzpicture}
\fill (0,0) circle (2pt);
\fill (-1,-1.5) circle (2pt);
\fill (-2,-3) circle (2pt);
\fill (0,-3) circle (2pt);
\fill (2,-3) circle (2pt);
\draw [thick] (0,0) -- (-1,-1.5);
\draw [thick] (-1,-1.5) -- (-2,-3);
\draw [thick](0,0) -- (2,-3);
\draw [thick](-1,-1.5) -- (0,-3);
\draw [thick](2.2,-1.5) -- (2.7,-1.5);
\draw [thick](2.5,-1.3) -- (2.7,-1.5);
\draw [thick](2.5,-1.7) -- (2.7,-1.5);
\draw (-2,-3.3) node {$a$};      
\draw (0,-3.3) node {$b$};      
\draw (2,-3.3) node {$c$};      
\end{tikzpicture}}
\subfigure[Attachment points added]{
\begin{tikzpicture}
\fill (0,0) circle (2pt);
\fill (-1,-1.5) circle (2pt);
\fill (-2,-3) circle (2pt);
\fill (0,-3) circle (2pt);
\fill (2,-3) circle (2pt);
\fill (0.4987,-0.7480) circle (2pt);
\fill (0.7979,-1.1968) circle (2pt);
\fill (1.2965,-1.9448) circle (2pt);
\fill (1.5957,-2.3936) circle (2pt);
\fill (-0.4987,-0.7480) circle (2pt);
\fill (-0.7979,-1.1968) circle (2pt);
\fill (-1.1968,-1.7952) circle (2pt);
\fill (-1.5957,-2.3936) circle (2pt);
\fill (-0.7979,-1.8032) circle (2pt);
\fill (-0.5984,-2.1024) circle (2pt);
\fill (-0.2992,-2.5512) circle (2pt);
\draw [thick](0,0) -- (-1,-1.5);
\draw [thick](-1,-1.5) -- (-2,-3);
\draw [thick](0,0) -- (2,-3);
\draw [thick](-1,-1.5) -- (0,-3);
\draw [thick](2.2,-1.5) -- (2.7,-1.5);
\draw [thick](2.5,-1.3) -- (2.7,-1.5);
\draw [thick](2.5,-1.7) -- (2.7,-1.5);
\draw (-2,-3.3) node {$a$};      
\draw (0,-3.3) node {$b$};      
\draw (2,-3.3) node {$c$};      
\end{tikzpicture}}
\subfigure[Linking arcs added]{
\begin{tikzpicture}
\fill (0,0) circle (2pt);
\fill (-1,-1.5) circle (2pt);
\fill (-2,-3) circle (2pt);
\fill (0,-3) circle (2pt);
\fill (2,-3) circle (2pt);
\fill (0.4987,-0.7480) circle (2pt);
\fill (0.7979,-1.1968) circle (2pt);
\fill (1.2965,-1.9448) circle (2pt);
\fill (1.5957,-2.3936) circle (2pt);
\fill (-0.4987,-0.7480) circle (2pt);
\fill (-0.7979,-1.1968) circle (2pt);
\filldraw ([xshift=-2pt,yshift=-2pt]-1.1968,-1.7952) rectangle ++(4pt,4pt);
\fill (-1.5957,-2.3936) circle (2pt);
\fill (-0.7979,-1.8032) circle (2pt);
\filldraw ([xshift=-2pt,yshift=-2pt]-0.5984,-2.1024) rectangle ++(4pt,4pt);
\fill (-0.5984,-2.1024) circle (2pt);
\fill (-0.2992,-2.5512) circle (2pt);
\draw [thick](0,0) -- (-1,-1.5);
\draw [thick](-1,-1.5) -- (-2,-3);
\draw [thick](0,0) -- (2,-3);
\draw [thick](-1,-1.5) -- (0,-3);
\draw [thick](2.2,-1.5) -- (2.7,-1.5);
\draw [thick](2.5,-1.3) -- (2.7,-1.5);
\draw [thick](2.5,-1.7) -- (2.7,-1.5);
\draw [thick](0.4987,-0.7480) -- (-0.7979,-1.1968);
\draw [thick](0.7979,-1.1968) -- (-0.7979,-1.8032);
\draw [thick](1.2965,-1.9448) -- (-0.2992,-2.5512); 
\draw [thick](-0.5984,-2.1024) -- (-1.5957,-2.3936);
\draw (-2,-3.3) node {$a$};      
\draw (0,-3.3) node {$b$};      
\draw (2,-3.3) node {$c$};      
\end{tikzpicture}}
\subfigure[Attachment points~suppressed]{
\begin{tikzpicture}
\fill (0,0) circle (2pt);
\filldraw ([xshift=-2pt,yshift=-2pt]-1,-1.5) rectangle ++(4pt,4pt);
\fill (-1,-1.5) circle (2pt);
\fill (-2,-3) circle (2pt);
\fill (0,-3) circle (2pt);
\fill (2,-3) circle (2pt);
\fill (0.4987,-0.7480) circle (2pt);
\fill (0.7979,-1.1968) circle (2pt);
\fill (1.2965,-1.9448) circle (2pt);
\fill (-0.7979,-1.1968) circle (2pt);

\fill (-1.5957,-2.3936) circle (2pt);
\fill (-0.7979,-1.8032) circle (2pt);
\filldraw ([xshift=-2pt,yshift=-2pt]-0.5984,-2.1024) rectangle ++(4pt,4pt);
\fill (-0.5984,-2.1024) circle (2pt);
\fill (-0.2992,-2.5512) circle (2pt);
\draw [thick] (0,0) -- (-1,-1.5);
\draw [thick](-1,-1.5) -- (-2,-3);
\draw [thick](0,0) -- (2,-3);
\draw [thick](-1,-1.5) -- (0,-3);
\draw [thick](2.2,-1.5) -- (2.7,-1.5);
\draw [thick](2.5,-1.3) -- (2.7,-1.5);
\draw [thick](2.5,-1.7) -- (2.7,-1.5);
\draw [thick](0.4987,-0.7480) -- (-0.7979,-1.1968);
\draw [thick](0.7979,-1.1968) -- (-0.7979,-1.8032);
\draw [thick](1.2965,-1.9448) -- (-0.2992,-2.5512); 
\draw [thick](-0.5984,-2.1024) -- (-1.5957,-2.3936);
\draw (-2,-3.3) node {$a$};      
\draw (0,-3.3) node {$b$};      
\draw (2,-3.3) node {$c$};      
\end{tikzpicture}}
\subfigure[\leo{Redrawn as the network from Figure \ref{Examplenetwork}}]{\begin{tikzpicture}[scale=0.85]
\filldraw ([xshift=-2.5pt,yshift=-2.5pt]-1.1,2.75) rectangle ++(5pt,5pt);
\filldraw ([xshift=-2.5pt,yshift=-2.5pt]-1.125,3.5) rectangle ++(5pt,5pt);
\fill (0,5) circle (2.5pt); 
\fill (-0.75,4) circle (2.5pt); 
\fill (0.375,4.5) circle (2.5pt); 
\fill (1.125,3.5) circle (2.5pt);
\fill (1.5,3) circle (2.5pt); 
\fill (-1.875,2.5) circle (2.5pt); 
\fill (2.25,2) circle (2.5pt); 
\fill (-0.035,3.1) circle (2.5pt);
\fill (-0.5,2) circle (2.5pt);
\fill (-0.5,1.5) circle (2.5pt); 
\fill (-2.4375,1.75) circle (2.5pt); 
\draw [thick] (0,5)  -- (0.75,4) ; 
\draw [thick] (0,5)  -- (-0.75,4) ; 
\draw [thick] (-0.75,4)  -- (-1.5,3) ; 
\draw [thick] (0.75,4)  -- (1.5,3) ; 
\draw [thick] (1.5,3)  -- (2.25,2) ; 
\draw [thick] (-1.5,3)  -- (-2.4375,1.75) ; 
\draw [thick] (0.375,4.5) -- (-0.75,4);
\draw [thick] (1.125,3.5) -- (-0.035,3.1);
\draw [thick] (-0.035,3.1) -- (-1.875,2.5);
\draw [thick] (-1.125,3.5) -- (-0.03,3.1);
\draw [thick] (1.5,3) -- (-0.5,2);
\draw [thick] (-0.5,2) -- (-0.5,1.5);
\draw [thick]  (-1.1,2.75) -- (-0.5,2);
\draw (-2.4375,1.45) node {$a$};      
\draw (-0.5,1.15) node {$b$};      
\draw (2.25,1.7) node {$c$};      
\end{tikzpicture}}
\caption{From phylogenetic tree to phyogenetic network in steps (a) to (e), which shows that the network from Figure~\ref{Examplenetwork} is tree-based.}
\label{definitionTB}
\end{figure}

A \emph{rooted spanning tree} $\tau$ of a phylogenetic network~$N$ is a subgraph of~$N$ that is a rooted tree and contains all vertices (and a subset of the arcs) of $N$. 
A \emph{dummy leaf} of a rooted spanning tree~$\tau$ is a vertex that is not a leaf in network $N$, but is a leaf in $\tau$. Hence, a binary phylogenetic network is tree-based if and only if it has a rooted spanning tree without dummy leaves.

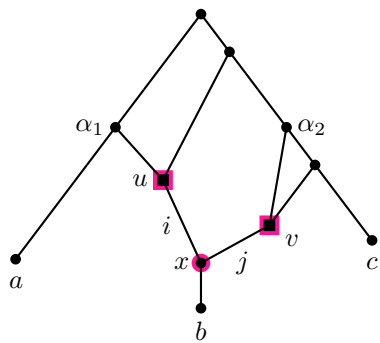
\begin{figure}[h]\centering
\begin{tikzpicture}
\fill (0,5) circle (2pt); 
\fill (0.375,4.5) circle (2pt); 
\fill (1.125,3.5) circle (2pt);
\fill (-1.125,3.5) circle (2pt); 
\fill (1.5,3) circle (2pt); 
\fill (2.25,2) circle (2pt); 
\filldraw [deeppink] ([xshift=-3.5pt,yshift=-3.5pt]-0.5,2.8) rectangle ++(7pt,7pt);
\filldraw ([xshift=-2pt,yshift=-2pt]-0.5,2.8) rectangle ++(4pt,4pt);
\fill[color=deeppink] (0,1.7) circle (3.5pt);
\fill (0,1.7) circle (2pt);
\fill (0,1.1) circle (2pt);
\filldraw [deeppink] ([xshift=-3.5pt,yshift=-3.5pt]0.9,2.2) rectangle ++(7pt,7pt);
\filldraw ([xshift=-2pt,yshift=-2pt]0.9,2.2) rectangle ++(4pt,4pt);
\fill (-2.4375,1.75) circle (2pt);
\draw [thick] (0,5)  -- (0.75,4) ; 
\draw [thick] (0,5)  -- (-0.75,4) ; 
\draw [thick] (-0.75,4)  -- (-1.5,3) ; 
\draw [thick] (0.75,4)  -- (1.125,3.5);
\draw [thick] (1.125,3.5) -- (1.5,3) ; 
\draw [thick] (1.125,3.5) -- (0.9,2.2);
\draw [thick] (1.5,3)  -- (2.25,2) ; 
\draw [thick] (-1.5,3)  -- (-2.4375,1.75) ; 
\draw [thick] (0.375,4.5) -- (-0.5,2.8);
\draw [thick] (-0.5,2.8) -- (0,1.7);
\draw [thick] (-1.125,3.5) -- (-0.5,2.8);
\draw [thick] (1.5,3) -- (0.9,2.2);
\draw [thick] (0.9,2.2) -- (0,1.7);
\draw [thick] (0,1.7) -- (0,1.1);
\draw (-2.43,1.45) node {$a$};      
\draw (0,0.8) node {$b$};      
\draw (-0.45,2.2) node {$i$};      
\draw (0.55,1.7) node {$j$};      
\draw (-0.8,2.8) node {$u$};      
\draw (1.2,2.0) node {$v$};      
\draw (2.25,1.7) node {$c$};      
\draw (-0.25,1.7) node {$x$};      
\draw (-1.46,3.5) node {$\alpha_1$};      
\draw (1.46,3.5) node {$\alpha_2$};      
\end{tikzpicture}
\caption{An example of a non-tree-based binary phylogenetic network. Because arcs~$i$ and~$j$ are the only outgoing arcs of vertices~$u$ and~$v$ respectively, they would both have to be present in the base-tree. However, vertex~$x$ would then have two incoming arcs in the base-tree, which is not allowed.}
\label{ExamplenetworknotTB}
\end{figure}

We will make heavy use of the following bipartite graph. Let $N=(V,A)$ be a binary phylogenetic network. The \emph{bipartite graph associated to} $N$ is the bipartite graph $B=(U \cup R,E)$ containing \leo{a vertex $v_o\in U$ for each~$v\in V$ that is an omnian, a \rev{vertex~$w_r\in R$ for each~$w\in V$} that is a reticulation, and an edge~$\{v_o,w_r\}\in E$ for each $(v,w) \in A$ with~$v$ and omnian and~$w$ a reticulation. Hence, for a vertex~$v\in V$ that is a reticulation and an omnian, there is a vertex~$v_o$ in~$U$ as well as a vertex~$v_r$ in~$R$. For ease of notation, we will omit the subscripts for now on and refer to~$v_o$ and~$v_r$ simply as $v$.} An example is given in Figure~\ref{VBmatch}.


Let $N=(V,A)$ be a binary phylogenetic network. An \emph{antichain} is a set of vertices~$K~\subseteq~V$ for which there is no directed path from \leo{any} vertex in~$K$ to \leo{any other} vertex in~$K$. Network~$N$ satisfies 
the \emph{antichain-to-leaf property} if for every antichain in $N$ there exists a path from every vertex in $K$ to a leaf, so that these paths are arc-disjoint. 
Which means, for example, that if there is an antichain of three vertices and there are only two leaves in the network, the network does clearly not satisfy the antichain-to-leaf property.

An example of an antichain can be seen in Figure \ref{ExamplenetworknotTB}, where vertices $\alpha_1$ and $\alpha_2$ form an antichain. 
The network does not satisfy the antichain-to-leaf property, because when we look at the antichain formed by vertices $u$ and $v$, there are no arc-disjoint paths to leaves.

A vertex $v$  is called \emph{stable} if there exists a leaf $l$ for which every path from the root to $l$ passes through $v$. A network is called \emph{stable} if every reticulation is stable.

Let $G=(V,E)$ be a graph. If $v,w \in V$ so that $(v,w) \in E$, then $w$ is a \emph{neighbour} of $v$. For a set $S \subseteq V$, the neighbours of $S$ are denoted by $\Gamma(S)$.
A \emph{matching} $M \subseteq E$ is a set of edges so that no vertex $v \in V$ is incident with more than one edge in $M$.
A \emph{maximal path} in $G$ is a path that is not contained in a larger path.


%

The following known results will be useful.

\begin{proposition}\cite{Artikel2}
\label{Propositie3i)} 
Consider a binary phylogenetic network N over leaf set X.
\begin{enumerate}
\item[(i)] If the parents of each reticulation of $N$ are tree-vertices, then $N$ is tree-based.
\item[(ii)] If $N$ has a reticulation whose parents are both reticulations, then $N$ is not tree-based.
\end{enumerate}
\end{proposition}

\begin{proposition} 
\label{Propositie4.1}\label{4.1} \cite{Artikel3}
In a stable binary phylogenetic network $N$, the child and the parents of each reticulation are tree-vertices.
\end{proposition}

The following observation follows directly from the previous two propositions.

\begin{corollary} 
\label{Corollary}
Every binary stable phylogenetic network is tree-based.
\end{corollary}
%

\subsection{Results}
\label{subsec:binresults}
The following theorem will be used to obtain a simple graph-theoretic characterization of binary tree-based phylogenetic networks.

\begin{theorem}
\label{Lstelling}
Let~$N$ be a binary phylogenetic network and $B=(U \cup R,E)$ the bipartite graph associated to~$N$. Network~$N$ is tree-based  if and only if there exists a matching~$M$~in~$B$ with $|M|=|U|$.
\end{theorem}
\begin{proof} Assume there exists a matching $M$ in $B$ with $|M|=|U|$, i.e., all omnians are covered by $M$. 
Construct a set~\leo{$A$} of arcs as follows: \leo{add} the outgoing arc of every reticulation and the incoming arc of all tree-vertices to $A$.
Additionally, for each edge of $M$, add the corresponding arc of~$N$ to $A$, if it has not yet been added. 
For every reticulation that has not yet been covered, add one of its incomming arcs to $A$, arbitrarily. 
The tree~$T$, consisting of all vertices of $N$ and the set of arcs $A$, is a rooted spanning tree, because there is precisely one incoming arc of every vertex contained in $T$. Moreover, there are no dummy leaves, because $U$ is covered by~$M$. Hence, it follows that~$N$ is tree-based.

Now, assume that $N$ is tree-based with base-tree~$T$. Colour every edge of~$B$ that corresponds to an arc in $T$. 
When an omnian has outdegree 2 and both outgoing arcs are contained in $T$, decolourize one of the two corresponding edges of~$B$, arbitrarily. Hence, each vertex of~$U$ is incident to at most one coloured edge. Since~$T$ is a rooted tree, it contains at most one incoming arc of each reticulation. Hence, also each vertex of~$R$ is incident to at most one coloured edge. So the coloured edges of~$B$ form a matching~$M$.
Because~$T$ is a base-tree, there are no dummy leaves, and so all omnians are covered by~$M$.
\end{proof}

This theorem can be used to verify whether a binary phylogenetic network $N$ is tree-based or not in polynomial time, using an algorithm for maximum cardinality bipartite matching (see e.g.~\cite{schrijver}).

We will look at an example of a binary phylogenetic network $N$ and the bipartite graph $B=(U \cup R,E)$ associated to $N$ in Figure \ref{VBmatch}. Since there exists a matching, which is coloured blue and \leo{dash-dotted} in Figure \ref{VBmatch}(b), that covers $U$, the binary phylogenetic network in Figure \ref{VBmatch}(a) is tree-based.
A base-tree~$T$ of network~$N$ can be seen in Figure \ref{Basetree}, where the arcs that correspond to edges of the matching are \leo{dash-dotted} \leo{and linking-arcs are dashed}.

\begin{figure}[h]\centering
\subfigure[A rooted binary phylogenetic network~$N$.]{\begin{tikzpicture}

\fill (0,0) circle (2pt);
\fill [deeppink] (-1,-2) circle (3.5pt);
\fill [deeppink] (1,-2) circle (3.5pt);
\fill [deeppink] (-1.5,-2.5) circle (3.5pt);
\fill [deeppink] (1.5,-2.5) circle (3.5pt);
\fill (-1,-1) circle (2pt);
\fill (1.5,-1.5) circle (2pt);
\fill (-1.5,-1.5) circle (2pt);
\filldraw ([xshift=-2pt,yshift=-2pt]0,-1.5) rectangle ++(4pt,4pt);
\fill (2,-2) circle (2pt);
\fill (-2,-2) circle (2pt);
\filldraw [deeppink] ([xshift=-3.5pt,yshift=-3.5pt]1,-2) rectangle ++(7pt,7pt);
\filldraw ([xshift=-2pt,yshift=-2pt]1,-2) rectangle ++(4pt,4pt);
\fill (-1.5,-2.5) circle (2pt);
\fill (1.5,-2.5) circle (2pt);
\fill (-1.5,-3) circle (2pt);
\fill (1.5,-3) circle (2pt);
\filldraw [deeppink] ([xshift=-3.5pt,yshift=-3.5pt]-1,-2) rectangle ++(7pt,7pt);
\filldraw ([xshift=-2pt,yshift=-2pt]-1,-2) rectangle ++(4pt,4pt);
\fill (-2.2,-2.8) circle (2pt);
\fill (2.2,-2.8) circle (2pt);
\draw [thick] (0,0) -- (1.5,-1.5);
\draw [thick](0,0) -- (-1,-1);
\draw [thick](-1,-1) -- (0,-1.5);
\draw [thick](-1,-1) -- (-1.5,-1.5);
\draw [thick](1.5,-1.5) -- (2,-2);
\draw [thick](0,-1.5) -- (-1,-2);
\draw [thick](0,-1.5) -- (1,-2);
\draw [thick](-1.5,-1.5) -- (-2,-2);
\draw [thick](-1,-2) -- (-1.5,-2.5);
\draw [thick](-1.5,-2.5) -- (-1.5, -3);
\draw [thick](1,-2) -- (1.5,-2.5);
\draw [thick](1.5,-2.5) -- (1.5,-3);
\draw [thick](2,-2) -- (1.5,-2.5);
\draw [thick](-2,-2) -- (-1.5,-2.5);
\draw [thick](1.5,-1.5) -- (1,-2);
\draw [thick](-1.5,-1.5) -- (-1,-2);
\draw [thick](2.2,-2.8) -- (2,-2);
\draw [thick](-2.2,-2.8) -- (-2,-2);
\draw (-0.7,-2.1) node {$y$};
\draw (0.7,-2.1) node {$z$};
\draw (0.1,-1.3) node {$v$};
\draw (-1.2,-2.5) node {$q$};
\draw (1.2,-2.5) node {$x$};
\draw (-2.2,-3.1) node {$a$};
\draw (-1.5,-3.3) node {$b$};
\draw (2.2,-3.1) node {$d$};
\draw (1.5,-3.3) node {$c$};
\end{tikzpicture}
}
\hspace{2cm}
\subfigure[The bipartite graph associated to~$N$, with a matching covering~$U$ indicated by \leo{dash-dotted} lines.]{
\begin{tikzpicture}
\draw [blue, dashdotted, thick] (0,0) -- (2.5,-2.4);
\draw [blue, dashdotted, thick] (0,-1.2) -- (2.5,-1.6);
\draw [blue, dashdotted, thick] (0,-2.4) -- (2.5,0);
\draw [thick] (0,-2.4) -- (2.5,-0.8);
\draw (0,1) node {$U$};
\draw (2.5,1) node {$R$};
\draw (-0.3,0.0) node {$y$};
\draw (-0.3,-1.2) node {$z$};
\draw (-0.3,-2.4) node {$v$};
\draw (2.8,0.0) node {$y$};
\draw (2.8,-0.8) node {$z$};
\draw (2.8,-1.6) node {$\leo{x}$};
\draw (2.8,-2.4) node {$\leo{q}$};

\filldraw [deeppink] ([xshift=-3.5pt,yshift=-3.5pt]0,-1.2) rectangle ++(7pt,7pt);
\filldraw ([xshift=-2pt,yshift=-2pt]0,-1.2) rectangle ++(4pt,4pt);
\filldraw [deeppink] ([xshift=-3.5pt,yshift=-3.5pt]2.5,0) rectangle ++(7pt,7pt);
\filldraw ([xshift=-2pt,yshift=-2pt]2.5,0) rectangle ++(4pt,4pt);
\filldraw [deeppink] ([xshift=-3.5pt,yshift=-3.5pt]2.5,-0.8) rectangle ++(7pt,7pt);
\filldraw ([xshift=-2pt,yshift=-2pt]2.5,-0.8) rectangle ++(4pt,4pt);
\fill [deeppink] (2.5,-1.6) circle (3.5pt);
\fill [deeppink] (2.5,-2.4) circle (3.5pt);
\filldraw [deeppink] ([xshift=-3.5pt,yshift=-3.5pt]0,0) rectangle ++(7pt,7pt);
\filldraw ([xshift=-2pt,yshift=-2pt]0,0) rectangle ++(4pt,4pt);
\filldraw ([xshift=-2pt,yshift=-2pt]0,-2.4) rectangle ++(4pt,4pt);
\fill (2.5,-1.6) circle (2pt);
\fill (2.5,-2.4) circle (2pt);
\end{tikzpicture}
}
\caption{Using Theorem \ref{Lstelling} to show that this is a tree-based binary phylogenetic network.}
\label{VBmatch}
\end{figure}
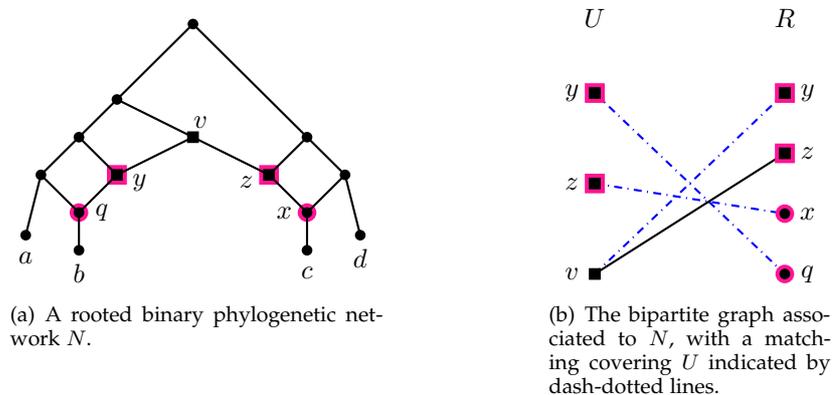

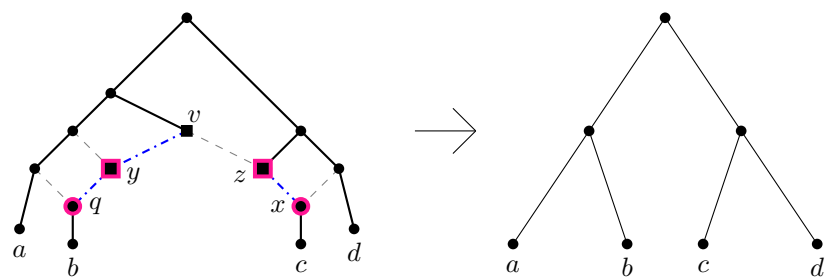
\begin{figure}[h]\centering
\subfigure{
\begin{tikzpicture}

\draw [thick] (0,0) -- (1.5,-1.5);
\draw [thick](0,0) -- (-1,-1);
\draw [thick](-1,-1) -- (0,-1.5);
\draw [thick](-1,-1) -- (-1.5,-1.5);
\draw [thick](1.5,-1.5) -- (2,-2);
\draw [thick, blue, dashdotted](0,-1.5) -- (-1,-2);
\draw [gray, dashed](0,-1.5) -- (1,-2);
\draw [thick](-1.5,-1.5) -- (-2,-2);
\draw [thick, blue, dashdotted](-1,-2) -- (-1.5,-2.5);
\draw [thick](-1.5,-2.5) -- (-1.5, -3);
\draw [thick, blue, dashdotted](1,-2) -- (1.5,-2.5);
\draw [thick](1.5,-2.5) -- (1.5,-3);
\draw [gray, dashed](2,-2) -- (1.5,-2.5);
\draw [gray, dashed](-2,-2) -- (-1.5,-2.5);
\draw [thick](1.5,-1.5) -- (1,-2);
\draw [gray, dashed](-1.5,-1.5) -- (-1,-2);
\draw [thick](2.2,-2.8) -- (2,-2);
\draw [thick](-2.2,-2.8) -- (-2,-2);
\draw (-0.7,-2.1) node {$y$};
\draw (0.7,-2.1) node {$z$};
\draw (0.1,-1.3) node {$v$};
\draw (-1.2,-2.5) node {$q$};
\draw (1.2,-2.5) node {$x$};
\draw (-2.2,-3.1) node {$a$};
\draw (-1.5,-3.3) node {$b$};
\draw (2.2,-3.1) node {$d$};
\draw (1.5,-3.3) node {$c$};
\draw (3,-1.5) -- (3.8,-1.5);
\draw (3.5,-1.2) -- (3.8,-1.5);
\draw (3.5,-1.8) -- (3.8,-1.5);

\fill (0,0) circle (2pt);
\fill [deeppink] (-1,-2) circle (3.5pt);
\fill [deeppink] (1,-2) circle (3.5pt);
\fill [deeppink] (-1.5,-2.5) circle (3.5pt);
\fill [deeppink] (1.5,-2.5) circle (3.5pt);
\fill (-1,-1) circle (2pt);
\fill (1.5,-1.5) circle (2pt);
\fill (-1.5,-1.5) circle (2pt);
\filldraw ([xshift=-2pt,yshift=-2pt]0,-1.5) rectangle ++(4pt,4pt);
\fill (2,-2) circle (2pt);
\fill (-2,-2) circle (2pt);
\filldraw [deeppink] ([xshift=-3.5pt,yshift=-3.5pt]1,-2) rectangle ++(7pt,7pt);
\filldraw ([xshift=-2pt,yshift=-2pt]1,-2) rectangle ++(4pt,4pt);
\fill (-1.5,-2.5) circle (2pt);
\fill (1.5,-2.5) circle (2pt);
\fill (-1.5,-3) circle (2pt);
\fill (1.5,-3) circle (2pt);
\filldraw [deeppink] ([xshift=-3.5pt,yshift=-3.5pt]-1,-2) rectangle ++(7pt,7pt);
\filldraw ([xshift=-2pt,yshift=-2pt]-1,-2) rectangle ++(4pt,4pt);
\fill (-2.2,-2.8) circle (2pt);
\fill (2.2,-2.8) circle (2pt);
\end{tikzpicture}
}
\subfigure{
\begin{tikzpicture}
\fill (0,0) circle (2pt);
\fill (-2,-3) circle (2pt);
\fill (2,-3) circle (2pt);
\fill (-1,-1.5) circle (2pt);
\fill (1,-1.5) circle (2pt);
\fill (0.5,-3) circle (2pt);
\fill (-0.5,-3) circle (2pt);
\draw (0,0) -- (-2,-3);
\draw (0,0) -- (2,-3);
\draw (-1,-1.5) -- (-0.5,-3);
\draw (1,-1.5) -- (0.5,-3);
\draw (-2,-3.3) node {$a$};
\draw (-0.5,-3.3) node {$b$};
\draw (2,-3.3) node {$d$};
\draw (0.5,-3.3) node {$c$};
\end{tikzpicture}
}
\caption{A base-tree $T$ of the network in~Figure~\ref{VBmatch}(a).}
\label{Basetree}
\end{figure}

Since a binary phylogenetic network that contains no reticulations is a rooted tree, such a network is clearly tree-based.  The next theorem shows that this is still the case for all networks with one or two reticulations. On the other hand, Figure~\ref{easynotTBexample} shows a part of a network $N$ that contains three reticulations and is not tree-based. So it follows that not all networks with three reticulations are tree-based.

\begin{theorem}\label{thm:2ret}
If a binary phylogenetic network $N$ contains at most two reticulations, then~$N$ is tree-based.
\end{theorem}

\begin{proof} If $N$ contains only one reticulation, then both parents of this reticulation are tree-vertices and with Proposition \ref{Propositie3i)} it follows that~$N$ is tree-based.

Now consider the case that~$N$ contains exactly two reticulations $x$ and $y$. If $x$ and $y$ do not have a parent-child relation, then both parents of $x$ and $y$ are tree-vertices and it follows from Proposition \ref{Propositie3i)} that $N$ is tree-based. Now suppose that $x$ is the parent of $y$. There are two possibilities, $x$ and $y$ having a joint parent and $x$ and $y$ having different parents, both displayed in Figure~\ref{Reen}.

\begin{figure}[h]\centering
\begin{tikzpicture}
\fill [deeppink] (1,-1) circle (3.5pt);
\filldraw [deeppink] ([xshift=-3.5pt,yshift=-3.5pt]0,0) rectangle ++(7pt,7pt);
\filldraw ([xshift=-2pt,yshift=-2pt]0,0) rectangle ++(4pt,4pt);
\filldraw [deeppink] ([xshift=-3.5pt,yshift=-3.5pt]2,0) rectangle ++(7pt,7pt);
\filldraw ([xshift=-2pt,yshift=-2pt]2,0) rectangle ++(4pt,4pt);
\fill (1,-1) circle (2pt);
\draw [thick] (-0.65,1) -- (0,0);
\draw [thick] (0.65,1) -- (0,0);
\draw [thick] (1.35,1) -- (2,0);
\draw [thick] (2.65,1) -- (2,0);
\draw [thick] (0,0) -- (1,-1);
\draw [thick] (1,-1) -- (2,0);
\draw [thick] (1,-1) -- (1,-2);

\end{tikzpicture}
\caption{Local situation in a network that has three reticulations and is not tree-based.}
\label{easynotTBexample}

\end{figure}
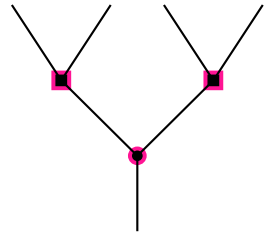 

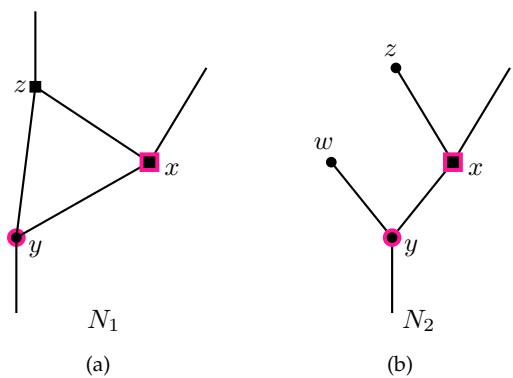
\begin{figure}[h]\centering
\subfigure[]{
\begin{tikzpicture}
\fill[color=deeppink] (-0.25,-2) circle (3.5pt);
\draw [thick] (0,0)  -- (1.5,-1) ;  
\draw [thick] (-0.25,-2) -- (-0.25,-3);   
\draw [thick] (0,0) -- (-0.25,-2);   
\draw [thick] (-0.25,-2) -- (1.5,-1);  
\draw [thick] (0,1) -- (0,0);    
\draw [thick] (2.25,0.25) -- (1.5,-1) ; 
\filldraw ([xshift=-2pt,yshift=-2pt]0,0) rectangle ++(4pt,4pt); 
\draw[color=black] (-0.2,0) node {$z$};    
\filldraw [color=deeppink] ([xshift=-3.5pt,yshift=-3.5pt]1.5,-1) rectangle ++(7pt,7pt);
\filldraw ([xshift=-2pt,yshift=-2pt]1.5,-1) rectangle ++(4pt,4pt);
\fill (1.5,-1) circle (2pt);
\draw[color=black] (1.8,-1.10) node {$x$};  
\fill (-0.25,-2) circle (2pt);
\draw[color=black] (0,-2.15) node {$y$};      
\draw[color=black] (0.9, -3.1) node {$N_1$};     
\end{tikzpicture}}
\hspace{1cm}
\subfigure[]{
\begin{tikzpicture}
\draw [thick] (6,-2) -- (6,-3);   
\draw [thick] (6.05,0.25)  -- (6.8,-1) ;  
\draw [thick] (5.2,-1) -- (6,-2);   
\draw [thick] (6,-2) -- (6.8,-1);  
\draw [thick] (7.55,0.25) -- (6.8,-1) ; 
\filldraw [color=deeppink] ([xshift=-3.5pt,yshift=-3.5pt]6.8,-1) rectangle ++(7pt,7pt);
\filldraw ([xshift=-2pt,yshift=-2pt]6.8,-1) rectangle ++(4pt,4pt);
\fill[color=deeppink] (6,-2) circle (3.5pt);
\fill (6.05,0.25) circle (2pt); 
\draw[color=black] (5.98,0.48) node {$z$};    
\fill (6.8,-1) circle (2pt);
\draw[color=black] (7.1,-1.10) node {$x$};  
\fill (6,-2) circle (2pt);
\draw[color=black] (6.25,-2.14) node {$y$};      
\fill (5.2,-1) circle (2pt);
\draw[color=black] (5.1,-.75) node {$w$};      
\draw[color=black] (6.35, -3.1) node {$N_2$};     
\end{tikzpicture}}
\caption{The two possibilities that can occur when reticulation $x$ is the parent of reticulation $y$, used in the proof of Theorem~\ref{thm:2ret}.}
\label{Reen}
\end{figure}

From \leo{partial networks~$N_1$ and~$N_2$} of Figure \ref{Reen} we create two bipartite graphs, $A=(U \cup R,E)$ associated to $N_1$ and $B=(U \cup R,E)$ associated to $N_2$, that are displayed in Figure~\ref{Rtwee}.

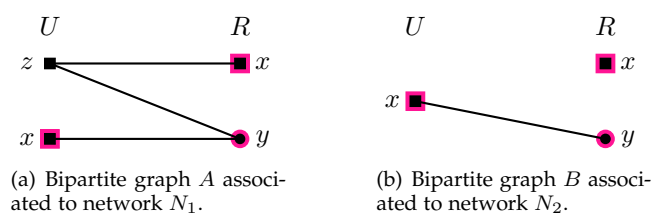
\begin{figure}[h]\centering
\subfigure[Bipartite graph $A$ associated to network~$N_1$.]{\begin{tikzpicture}
\draw[color=black] (0,0) node {$U$};      
\draw[color=black] (2.5,0) node {$R$};      
\filldraw ([xshift=-2pt,yshift=-2pt]0,-0.5) rectangle ++(4pt,4pt); 
\draw[color=black] (-0.3,-0.5) node {$z$};  
\filldraw [color=deeppink] ([xshift=-3.5pt,yshift=-3.5pt]0,-1.5) rectangle ++(7pt,7pt);
\filldraw ([xshift=-2pt,yshift=-2pt]0,-1.5) rectangle ++(4pt,4pt);   
\draw[color=black] (-0.3,-1.5) node {$x$};   
\filldraw [color=deeppink] ([xshift=-3.5pt,yshift=-3.5pt]2.5,-0.5) rectangle ++(7pt,7pt);
\filldraw ([xshift=-2pt,yshift=-2pt]2.5,-0.5) rectangle ++(4pt,4pt);    
\draw[color=black] (2.8,-0.5) node {$x$};      
\draw[color=black] (2.8,-1.5) node {$y$};    
\fill[color=deeppink] (2.5,-1.5) circle (3.5pt);
\fill (2.5,-1.5) circle (2pt);
\draw [thick] (0,-0.5) -- (2.5,-0.5);
\draw [thick] (0,-1.5) -- (2.5,-1.5); 
\draw [thick] (0,-0.5) -- (2.5,-1.5);
\end{tikzpicture}
}
\hspace{1cm}
\subfigure[Bipartite graph $B$ associated to network~$N_2$.]{\begin{tikzpicture}

\draw[color=black] (6,0) node {$U$};      
\draw[color=black] (8.5,0) node {$R$};      
\filldraw [color=deeppink] ([xshift=-3.5pt,yshift=-3.5pt]8.5,-0.5) rectangle ++(7pt,7pt);
\filldraw ([xshift=-2pt,yshift=-2pt]8.5,-0.5) rectangle ++(4pt,4pt);
\draw[color=black] (5.7,-1) node {$x$};          
\filldraw [color=deeppink] ([xshift=-3.5pt,yshift=-3.5pt]6,-1) rectangle ++(7pt,7pt);
\filldraw ([xshift=-2pt,yshift=-2pt]6,-1) rectangle ++(4pt,4pt);
\draw[color=black] (8.8,-0.5) node {$x$};      
\draw[color=black] (8.8,-1.5) node {$y$};  
\fill[color=deeppink] (8.5,-1.5) circle (3.5pt);
\fill (8.5,-1.5) circle (2pt);
\draw [thick] (6,-1) -- (8.5,-1.5);
\end{tikzpicture}
}
\caption{The bipartite graphs associated to the partial networks in~Figure~\ref{Reen}.}
\label{Rtwee}
\end{figure}

In both cases in Figure~\ref{Rtwee} it is easy to see that there is a matching that covers $U$. It then follows from Theorem~\ref{Lstelling} that $N$ is tree-based.
\end{proof}

To obtain a simple characterization of binary tree-based networks, we will use Hall's Theorem, which is stated below.

\begin{theorem}[Hall's Theorem~\cite{hall}]\label{Hall}
Let $B=(U\cup W,E)$ be a bipartite graph. 
There exists a matching in~$B$ that covers $U$ if and only if, for every~$U_1~\subseteq~U$, the number of different neighbours of the vertices in $U_1$ is at least $|U_1|$.
\end{theorem}

Consider Hall's Theorem and Theorem \ref{Lstelling}. Combining those two theorems gives a characterization for a binary phylogenetic network to be tree-based.
 
\begin{corollary}\label{Aantalkind}
Let $N$ be a binary phylogenetic network and $U$ the set of all omnians of~$N$. 
Then $N$ is tree-based if and only if  for all $S \subseteq U$ the number of different children of \leo{the vertices in}~$S$ is greater than or equal to the number of omnians in $S$.
\end{corollary}

\begin{proof}
Follows directly from Theorem \ref{Lstelling} and Theorem \ref{Hall}.
\end{proof}

An example of how this theorem and corollary can be applied is given in Figure \ref{VBnomatch}, where an example of a binary phylogenetic network $N$ is displayed in (a) and the bipartite graph~$B=(U~\cup~R,E)$ associated to $N$ in (b). \leo{Omnians are indicated as square nodes and reticulations are marked with a pink shading around the nodes.}

From the bipartite graph in Figure \ref{VBnomatch} it follows with Hall's Theorem, with $S = U$, that there exists no matching in $B$ that covers $U$. 
Therefore, with Theorem~\ref{Lstelling} it follows that $N$ in Figure \ref{VBnomatch}(a) is not tree-based. 
Indeed, we can directly see in~$N$ that the \leo{omians in} $S=\{a,i,h,f,g\}$ \leo{have} only four different children~$\{b,c,d,e\}$ \leo{(note that, in general, the set~$S$ and the set of children of vertices in~$S$ do not have to be disjoint).}
Hence this network is not tree-based.
\begin{figure}[h]\centering
\subfigure{\begin{tikzpicture}
\fill (0,0) circle (2pt);
\fill (-1,-1) circle (2pt);
\fill (1.8,-0.8) circle (2pt);
\fill (-2.5,-3) circle (2pt);
\fill (-2.8,-3.6) circle (2pt);
\draw (-2.8,-3.9) node {$p$};
\fill (-1.5,-5) circle (2pt);
\fill (-1.1,-5.5) circle (2pt);
\draw (-1.1,-5.8) node {$q$};
\filldraw [color=deeppink] ([xshift=-3.5pt,yshift=-3.5pt]-2,-6) rectangle ++(7pt,7pt);
\filldraw ([xshift=-2pt,yshift=-2pt]-2,-6) rectangle ++(4pt,4pt);
\draw (-2.3,-6) node {$f$};
\fill (-2,-6) circle (2pt);
\fill [deeppink] (-1.3,-7.5) circle (3.5pt);
\draw (-1.6,-7.5) node {$e$};
\fill (-1.3,-7.5) circle (2pt);
\fill (-1.3,-7.9) circle (2pt);
\draw (-1.3,-8.2) node {$w$};
\filldraw [color=deeppink] ([xshift=-3.5pt,yshift=-3.5pt]3,-2) rectangle ++(7pt,7pt);
\filldraw ([xshift=-2pt,yshift=-2pt]3,-2) rectangle ++(4pt,4pt);
\draw (3.3,-2) node {$a$};
\fill [deeppink] (2.5,-3) circle (3.5pt);

\fill (2.5,-3) circle (2pt);
\draw (2.8,-3) node {$b$};
\fill (0.5,-1.2) circle (2pt);
\filldraw ([xshift=-2pt,yshift=-2pt]0,-2.2) rectangle ++(4pt,4pt);
\draw (-0.3,-2.5) node {$i$};
\fill [deeppink] (0.5,-5.5) circle (3.5pt);

\fill (0.5,-5.5) circle (2pt);
\draw (0.2,-5.3) node {$c$};
\fill (2,-3.5) circle (2pt);
\fill (2.5,-4) circle (2pt);
\draw (2.5,-4.3) node {$x$};

\filldraw ([xshift=-2pt,yshift=-2pt]1,-4.5) rectangle ++(4pt,4pt);
\draw (0.7,-4.1) node {$h$};
\fill [deeppink] (1.5,-7.5) circle (3.5pt);
\fill (1.5,-7.5) circle (2pt);
\draw (1.8,-7.5) node {$d$};
\fill (0,-6) circle (2pt);

\filldraw ([xshift=-2pt,yshift=-2pt]-0.8,-6.5) rectangle ++(4pt,4pt);

\draw (-1,-6.3) node {$g$};
\fill (1.5,-7.9) circle (2pt);
\draw (1.5,-8.2) node {$d$};
\fill (0.2,-6.3) circle (2pt);
\draw (0.2,-6.6) node {$y$};

\draw [thick] (0,0) -- (-1,-1);
\draw [thick] (0,0) -- (1.8,-0.8);
\draw [thick] (-1,-1) -- (-2.5,-3);
\draw [thick] (-2.5,-3) -- (-2.8,-3.6);
\draw [thick] (-1.5,-5) -- (-1.1,-5.5);
\draw [thick] (-1.5,-5) -- (-2,-6);
\draw [thick] (-1.5,-5) -- (-1,-1);
\draw [thick] (-2.5,-3) -- (-2,-6);
\draw [thick] (-2,-6) -- (-1.3,-7.5);
\draw [thick] (-1.3,-7.9) -- (-1.3,-7.5);
\draw [thick] (1.8,-0.8) -- (3,-2);
\draw [thick] (3,-2) -- (2.5,-3);
\draw [thick] (0.5,-1.2) -- (3,-2);
\draw [thick] (0.5,-1.2) -- (1.8,-0.8);
\draw [thick] (0.5,-1.2) -- (0,-2.2);
\draw [thick] (0.5,-5.5) -- (0,-2.2);
\draw [thick] (2,-3.5) -- (2.5,-3);
\draw [thick] (2,-3.5) -- (1,-4.5);
\draw [thick] (0.5,-5.5) -- (1,-4.5);
\draw [thick] (1,-4.5) -- (1.5,-7.5);
\draw [thick] (0.5,-5.5) -- (0,-6);
\draw [thick] (-0.8,-6.5) -- (0,-6);
\draw [thick] (-0.8,-6.5) -- (-1.3,-7.5);
\draw [thick] (-0.8,-6.5) -- (1.5,-7.5);
\draw [thick] (1.5,-7.9) -- (1.5,-7.5);
\draw [thick] (2,-3.5) -- (2.5,-4);
\draw [thick] (0,-2.2) -- (2.5,-3);
\draw [thick] (0,-6) -- (0.2,-6.3);
\end{tikzpicture}
}
\hspace{2cm}
\subfigure{\begin{tikzpicture}
\filldraw [color=deeppink] ([xshift=-3.5pt,yshift=-3.5pt]0,0) rectangle ++(7pt,7pt);
\filldraw ([xshift=-2pt,yshift=-2pt]0,0) rectangle ++(4pt,4pt);
\filldraw [color=deeppink] ([xshift=-3.5pt,yshift=-3.5pt]0,-1) rectangle ++(7pt,7pt);
\filldraw ([xshift=-2pt,yshift=-2pt]0,-1) rectangle ++(4pt,4pt);
\filldraw ([xshift=-2pt,yshift=-2pt]0,-2) rectangle ++(4pt,4pt);
\filldraw ([xshift=-2pt,yshift=-2pt]0,-3) rectangle ++(4pt,4pt);
\filldraw ([xshift=-2pt,yshift=-2pt]0,-4) rectangle ++(4pt,4pt);

\filldraw [color=deeppink] ([xshift=-3.5pt,yshift=-3.5pt]2.5,0) rectangle ++(7pt,7pt);
\filldraw ([xshift=-2pt,yshift=-2pt]2.5,0) rectangle ++(4pt,4pt);
\filldraw [color=deeppink] ([xshift=-3.5pt,yshift=-3.5pt]2.5,-4) rectangle ++(7pt,7pt);
\filldraw ([xshift=-2pt,yshift=-2pt]2.5,-4) rectangle ++(4pt,4pt);
\fill [deeppink](2.5,-0.8) circle (3.5pt);
\fill [deeppink](2.5,-1.6) circle (3.5pt);
\fill [deeppink](2.5,-2.4) circle (3.5pt);
\fill [deeppink](2.5,-3.2) circle (3.5pt);
\fill (2.5,-0.8) circle (2pt);
\fill (2.5,-1.6) circle (2pt);
\fill (2.5,-2.4) circle (2pt);
\fill (2.5,-3.2) circle (2pt);

\draw (-.3,0) node {$f$};
\draw (-.3,-1) node {$a$};
\draw (-.3,-2) node {$i$};
\draw (-.3,-3) node {$h$};
\draw (-.3,-4) node {$g$};

\draw (2.8,0) node {$a$};
\draw (2.8,-0.8) node {$b$};
\draw (2.8,-1.6) node {$c$};
\draw (2.8,-2.4) node {$d$};
\draw (2.8,-3.2) node {$e$};
\draw (2.8,-4) node {$f$};

\draw (0,1) node {$U$};
\draw (2.5,1) node {$R$};
\draw [thick] (0,0) -- (2.5,-3.2);
\draw [thick] (0,-1) -- (2.5,-0.8);
\draw [thick] (0,-2) -- (2.5,-0.8);
\draw [thick] (0,-2) -- (2.5,-1.6);
\draw [thick] (0,-3) -- (2.5,-1.6);
\draw [thick] (0,-3) -- (2.5,-2.4);
\draw [thick] (0,-4) -- (2.5,-2.4);
\draw [thick] (0,-4) -- (2.5,-3.2);
\end{tikzpicture}
}
\caption{Example of a non-tree-based binary phylogenetic network $N$ and the bipartite graph $B$ associated to $N$.}
\label{VBnomatch}
\end{figure}
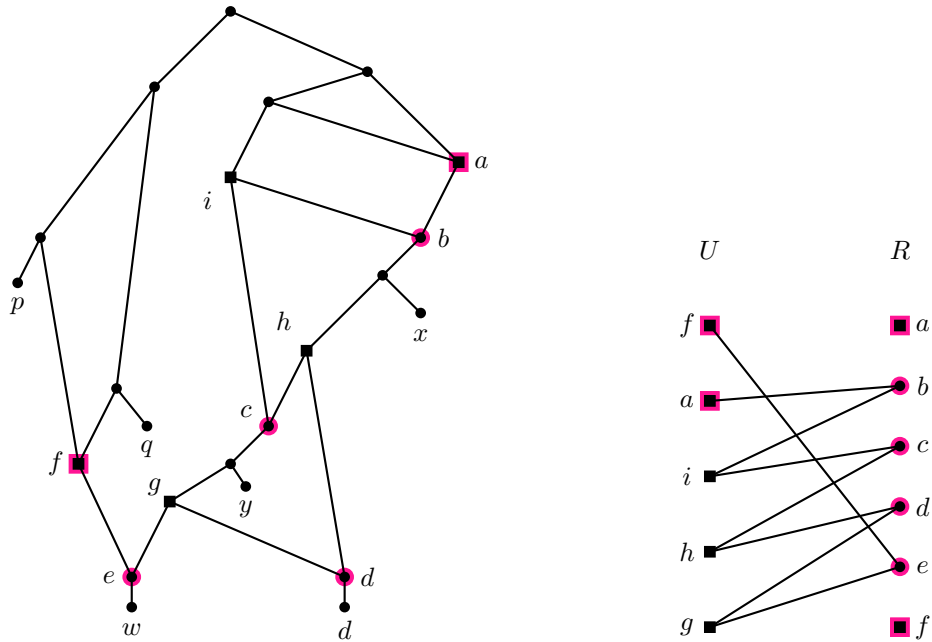

In a previous article, the following necessary condition for a network to be tree-based was found.

\begin{proposition}\cite{Artikel2} If a binary phylogenetic network over leaf set $X$ is tree-based, then it satisfies the antichain-to-leaf property.
\end{proposition}

On the other hand, if a network is not tree-based, it can still satisfy the antichain-to-leaf property, as shown by the example from~\cite{Artikel2} in Figure~\ref{WelALPgeenTB}. However, it is difficult to see in this network what is going on. \leo{Therefore, we show in Figure~\ref{ExamplesALP} two examples of local structures that cause a network to be non-tree-based. The example in Figure~\ref{ExamplesALP}(a) is similar to the local structure that causes the network of Figure~\ref{WelALPgeenTB} not to be tree-based, which can be checked using Corollary~\ref{Aantalkind}.}

At first sight, it might look like networks containing the local structures in Figure~\ref{ExamplesALP} cannot satisfy the antichain-to-leaf property. However, in Figure~\ref{Dpath} we show that it is indeed possible that they do satisfy this property. Although it was already known that networks that satisfy the antichain-to-leaf property are not necessarily tree-based, these figures illuminate why this is possible.

\begin{figure}[h]\centering
\begin{tikzpicture}[scale=0.67]
\fill (3,11) circle (3pt);  
\fill (5,9) circle (3pt);    
\fill (0,8) circle (3pt);    
\fill (7,7) circle (3pt);    
\filldraw ([xshift=-3pt,yshift=-3pt]1,6.2) rectangle ++(6pt,6pt);
\fill[color=deeppink] (5,5.66) circle (4.5pt);
\filldraw [deeppink] ([xshift=-5pt,yshift=-5pt]5,5.66) rectangle ++(10pt,10pt);
\filldraw ([xshift=-3pt,yshift=-3pt]5,5.66) rectangle ++(6pt,6pt);
\fill (1.5,3.33) circle (3pt);    
\fill[color=deeppink] (1,3) circle (4.5pt);
\filldraw [deeppink] ([xshift=-5pt,yshift=-5pt]1,3) rectangle ++(10pt,10pt);
\filldraw ([xshift=-3pt,yshift=-3pt]1,3) rectangle ++(6pt,6pt);
\fill[color=deeppink] (0,2.33) circle (4.5pt);
\fill[color=deeppink] (2.5,4) circle (4.5pt);
\fill (2.5,4) circle (3pt);
\fill (0,2.33) circle (3pt);    
\fill (8,6) circle (3pt);    
\fill (2,2.3) circle (3pt);    
\fill (-0.3,1.33) circle (3pt);    
\fill (1.5,3.33) circle (3pt);    

\draw (5.5,5.3) node {$o_1$};
\draw (2.9,3.6) node {$r_1$};
\draw (1.5,6.5) node {$o_2$};
\draw (-0.6,2.1) node {$r_2$};
\draw (1,2.5) node {$o_3$};

\draw [thick](3,11) -- (5,9);
\draw [thick](3,11) -- (0,8);
\draw [thick](5,9) -- (7,7);
\draw [thick](8,6) -- (7,7);
\draw [thick](0,8) -- (1,6.2);
\draw [thick](5,9) -- (5,5.66);
\draw [thick](5,5.66) -- (7,7);
\draw [thick](0,8) .. controls +(left:1cm) and +(down:1cm) ..  (1,3);   
\draw [thick](1,6.2) .. controls +(up:0.5cm)  and +(left:1 cm) .. (0,2.33);   
\draw [thick](5,5.66) -- (2.5,4);
\draw [thick](2.5,4) -- (1,6.2);
\draw [thick](2.5,4) -- (1.5,3.33);
\draw [thick](1.5,3.33) -- (2,2.3);
\draw [thick](1.5,3.33) -- (1,3);
\draw [ thick](1,3) -- (0,2.33);
\draw [thick](0,2.33) -- (-0.3,1.33);
\end{tikzpicture}
\caption{Not tree-based binary phylogenetic network satisfying the antichain-to-leaf property \cite{Artikel2}.}
\label{WelALPgeenTB}
\end{figure}
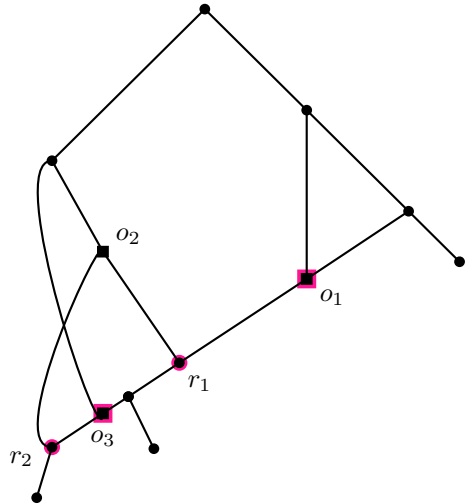

\begin{figure}[h]\centering
\hfill
\subfigure[]{
\begin{tikzpicture}
\draw [thick] (-0.65,1) -- (0,0);
\draw [thick] (0.65,1) -- (0,0);
\draw [thick] (2,1) -- (2,0);
\draw [thick, darkgreen, dashed] (2,0) -- (3,-1);
\draw [thick, darkgreen, dashed] (3,-1) -- (4,0);
\draw [thick] (4,0) -- (4,1);
\draw [thick] (3,-1) -- (3,-2);
\draw [thick, darkgreen, dashed] (5,-1) -- (4,0);
\draw [thick] (5,-1) -- (5,-2);
\draw [thick, darkgreen, dashed] (5,-1) -- (6,0);
\draw [thick, darkgreen, dashed] (6,0) -- (7,-1);
\draw [thick, darkgreen, dashed] (7,-1) -- (8,0);
\draw [thick] (7,-1) -- (7,-2);
\draw [thick] (6,0) -- (6,1);
\draw [thick] (8,0) -- (7.35,1);
\draw [thick] (8,0) -- (8.65,1);
\draw [thick, darkgreen, dashed] (0,0) -- (1,-1);
\draw [thick, darkgreen, dashed] (1,-1) -- (2,0);
\draw [thick] (1,-1) -- (1,-2);
\filldraw [color=deeppink] ([xshift=-3.5pt,yshift=-3.5pt]0,0) rectangle ++(7pt,7pt);
\filldraw ([xshift=-2pt,yshift=-2pt]0,0) rectangle ++(4pt,4pt);
\filldraw ([xshift=-2pt,yshift=-2pt]2,0) rectangle ++(4pt,4pt);
\fill[color=deeppink] (5,-1) circle (3.5pt);
\fill (5,-1) circle (2pt);
\fill[color=deeppink] (3,-1) circle (3.5pt);
\fill (3,-1) circle (2pt);
\fill[color=deeppink] (1,-1) circle (3.5pt);
\fill (1,-1) circle (2pt);
\filldraw ([xshift=-2pt,yshift=-2pt]4,0) rectangle ++(4pt,4pt);
\filldraw ([xshift=-2pt,yshift=-2pt]6,0) rectangle ++(4pt,4pt);
\fill[color=deeppink] (7,-1) circle (3.5pt);
\fill (7,-1) circle (2pt);
\filldraw [color=deeppink] ([xshift=-3.5pt,yshift=-3.5pt]8,0) rectangle ++(7pt,7pt);
\filldraw ([xshift=-2pt,yshift=-2pt]8,0) rectangle ++(4pt,4pt);
\draw (8.35,0) node {$o_5$};
\draw (6.35,0) node {$o_4$};
\draw (4.35,0) node {$o_3$};
\draw (2.35,0) node {$o_2$};
\draw (0.35,0) node {$o_1$};
\draw (1.35,-1) node {$r_1$};
\draw (3.35,-1) node {$r_2$};
\draw (5.35,-1) node {$r_3$};
\draw (7.35,-1) node {$r_4$};
\end{tikzpicture}}
\hfill
\subfigure[ ]{
\begin{tikzpicture}
\filldraw [color=deeppink] ([xshift=-3.5pt,yshift=-3.5pt]0,0) rectangle ++(7pt,7pt);
\filldraw ([xshift=-2pt,yshift=-2pt]0,0) rectangle ++(4pt,4pt);
\filldraw ([xshift=-2pt,yshift=-2pt]2,0) rectangle ++(4pt,4pt);
\fill[color=deeppink] (3,-1) circle (3.5pt);
\fill (3,-1) circle (2pt);
\filldraw [color=deeppink] ([xshift=-3.5pt,yshift=-3.5pt]1,-1) rectangle ++(7pt,7pt);
\filldraw ([xshift=-2pt,yshift=-2pt]1,-1) rectangle ++(4pt,4pt);
\fill (1,-1) circle (2pt);
\filldraw ([xshift=-2pt,yshift=-2pt]4,0) rectangle ++(4pt,4pt);
\fill[color=deeppink] (1.4,-3.2) circle (3.5pt);
\fill (1.4,-3.2) circle (2pt);
\draw [thick] (-0.65,1) -- (0,0);
\draw [thick] (0.65,1) -- (0,0);
\draw [thick] (2,1) -- (2,0);
\draw [thick] (2,0) -- (3,-1);
\draw [thick] (3,-1) -- (4,0);
\draw [thick] (4,0) -- (4,1);
\draw [thick] (3,-1) -- (3,-1.5);
\draw [thick] (0,0) -- (1,-1);
\draw [thick] (1,-1) -- (2,0);
\draw [thick] (1,-1) -- (1.4,-3.2);
\draw [thick] (1.4,-3.9) -- (1.4,-3.2);
\draw[thick] (4,0) .. controls +(right:1cm) and  +(right:1cm) .. (1.4,-3.2);
\draw (-.4,0) node {$o_1$};
\draw (1.7,0.2) node {$o_2$};
\draw (3.7,0.2) node {$o_3$};
\draw (0.15,-1.2) node {$r_1=o_4$};
\draw (2.7,-1.2) node {$r_2$};
\draw (1,-3.2) node {$r_3$};

\end{tikzpicture}}
\caption{Examples of local structures of binary phylogenetic networks that are not tree-based.}
\label{ExamplesALP}
\end{figure}
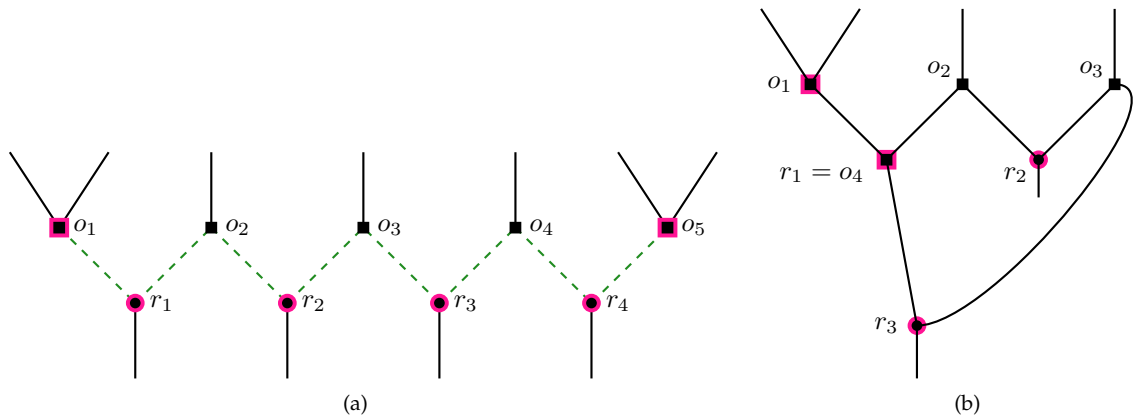

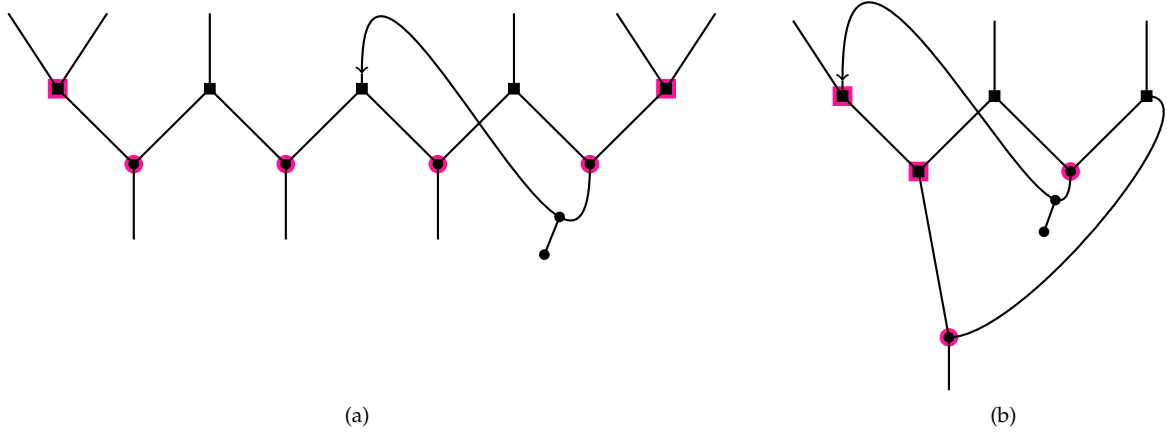
\begin{figure}[h]\centering
\subfigure[]{\begin{tikzpicture}
\filldraw [color=deeppink] ([xshift=-3.5pt,yshift=-3.5pt]0,0) rectangle ++(7pt,7pt);
\filldraw ([xshift=-2pt,yshift=-2pt]0,0) rectangle ++(4pt,4pt);
\filldraw ([xshift=-2pt,yshift=-2pt]2,0) rectangle ++(4pt,4pt);
\fill[color=deeppink] (5,-1) circle (3.5pt);
\fill (5,-1) circle (2pt);
\fill[color=deeppink] (3,-1) circle (3.5pt);
\fill (3,-1) circle (2pt);
\fill[color=deeppink] (1,-1) circle (3.5pt);
\fill (1,-1) circle (2pt);
\filldraw ([xshift=-2pt,yshift=-2pt]4,0) rectangle ++(4pt,4pt);
\filldraw ([xshift=-2pt,yshift=-2pt]6,0) rectangle ++(4pt,4pt);
\fill[color=deeppink] (7,-1) circle (3.5pt);
\fill (7,-1) circle (2pt);
\filldraw [color=deeppink] ([xshift=-3.5pt,yshift=-3.5pt]8,0) rectangle ++(7pt,7pt);
\filldraw ([xshift=-2pt,yshift=-2pt]8,0) rectangle ++(4pt,4pt);
\draw [thick] (-0.65,1) -- (0,0);
\draw [thick] (0.65,1) -- (0,0);
\draw [thick] (2,1) -- (2,0);
\draw [thick] (2,0) -- (3,-1);
\draw [thick] (3,-1) -- (4,0);
\draw [thick] (3,-1) -- (3,-2);
\draw [thick] (5,-1) -- (4,0);
\draw [thick] (5,-1) -- (5,-2);
\draw [thick] (5,-1) -- (6,0);
\draw [thick] (6,0) -- (7,-1);
\draw [thick] (7,-1) -- (8,0);

\draw [thick,->] (7,-1)  .. controls +(down:3cm) and +(up:3cm) .. (4,0.2);
\fill (6.6,-1.7) circle (2pt);
\draw [thick] (6.6,-1.7) -- (6.4,-2.2);
\fill (6.4,-2.2) circle (2pt);
\draw [thick] (4,0.2) -- (4,0);
\draw [thick] (6,0) -- (6,1);
\draw [thick] (8,0) -- (7.35,1);
\draw [thick] (8,0) -- (8.65,1);
\draw [thick] (0,0) -- (1,-1);
\draw [thick] (1,-1) -- (2,0);
\draw [thick] (1,-1) -- (1,-2);

\end{tikzpicture}}
\hfill
\subfigure[ ]{\begin{tikzpicture}
\filldraw [color=deeppink] ([xshift=-3.5pt,yshift=-3.5pt]0,0) rectangle ++(7pt,7pt);
\filldraw ([xshift=-2pt,yshift=-2pt]0,0) rectangle ++(4pt,4pt);
\filldraw ([xshift=-2pt,yshift=-2pt]2,0) rectangle ++(4pt,4pt);
\fill[color=deeppink] (3,-1) circle (3.5pt);
\fill (3,-1) circle (2pt);
\filldraw [color=deeppink] ([xshift=-3.5pt,yshift=-3.5pt]1,-1) rectangle ++(7pt,7pt);
\filldraw ([xshift=-2pt,yshift=-2pt]1,-1) rectangle ++(4pt,4pt);
\fill (1,-1) circle (2pt);
\filldraw ([xshift=-2pt,yshift=-2pt]4,0) rectangle ++(4pt,4pt);
\fill[color=deeppink] (1.4,-3.2) circle (3.5pt);
\fill (1.4,-3.2) circle (2pt);
\draw [thick] (-0.65,1) -- (0,0);
\draw [thick] (2,1) -- (2,0);
\draw [thick] (2,0) -- (3,-1);
\draw [thick] (3,-1) -- (4,0);
\draw [thick] (4,0) -- (4,1);
\draw [thick] (0,0) -- (1,-1);
\draw [thick] (1,-1) -- (2,0);
\draw [thick] (1,-1) -- (1.4,-3.2);
\draw [thick] (1.4,-3.9) -- (1.4,-3.2);
\draw[thick] (4,0) .. controls +(right:1cm) and  +(right:1cm) .. (1.4,-3.2);
\draw [thick,->] (3,-1) ..controls +(down:2cm) and +(up:3.5cm).. (0,0.2);
\draw [thick] (0,0.2) -- (0,0);
\fill (2.8,-1.38) circle (2pt);
\draw [thick] (2.8,-1.38) -- (2.65, -1.8);
\fill (2.65,-1.8) circle (2pt);
\end{tikzpicture}}
\caption{Local structures of binary phylogenetic networks that satisfy the antichain-to-leaf property but are not tree-based.}
\label{Dpath}
\end{figure}

Looking at the examples in Figure~\ref{ExamplesALP}, we see that a pattern has emerged. In (a) the pattern is marked dashed in green. 
Starting at vertex $o_1$ and ending at vertex $o_5$, we see a zigzag starting with an omnian, alternating between reticulations and omnians, eventually ending with an omnian.
The last omnian in the pattern can be a reticulation that is already part of the path, as can be seen in Figure~\ref{ExamplesALP}(b).

The next theorem shows that every binary phylogenetic network that is not tree-based contains a local structure as in the examples in Figure~\ref{ExamplesALP}.

\begin{theorem}
\label{Estelling2}
Let $N$ be a binary phylogenetic network and $B = (U\cup R,E)$ the bipartite graph associated to~$N$. Network~$N$ is tree-based if and only if~$B$ contains no maximal path which starts and ends in $U$.
\end{theorem}

\begin{proof}
Notice that every vertex in $B$ is of degree at most 2. Therefore,~$B$ is a disjoint union of paths and cycles. Hence, for each connected component~$B' = (U'\cup R',E')$ of~$B$, there are four \leo{possible topologies}:
\begin{enumerate}
\item[i)] A maximal path that begins and ends in~$R$.
\item[ii)] A maximal path that begins in $U$ and ends in $R$.
\item[iii)] A maximal path that begins and ends in $U$.
\item[iv)] A circuit.
\end{enumerate}

i) All vertices in $R$ are of degree at most 2. Because the maximal path begins and ends in $R$, all vertices in~$U'$ have degree $2$. Let $S \subseteq U'$. Recall that~$\Gamma(S)$ denotes the set of neighbours of vertices in~$S$.
The number of edges incident to $S = 2 \left| S \right| \leq$ the number of edges incident to $\Gamma (S) \leq 2 \left| \Gamma (S)\right|$.
Therefore, $|S|\leq |\Gamma(S)|$ for all~$S\subseteq U'$. It follows from Hall's Theorem that there exists a matching in $B'$ that covers $U'$.\\

ii) All vertices in $R$ are of degree at most 2. All vertices in~$U'$ have degree~$2$, except for the omnian~$o_1$ where the maximal path begins. Let $S \subseteq U'$. Consider the subgraph of~$B'$ induced by~$S\cup\Gamma(S)$. It consists of paths. First consider such a path that does not contain~$o_1$. Then the path must begin and end in~$R$, because every omnian in~$S$ except for~$o_1$ has two neighbours in~$\Gamma(S)$. Hence, the paths contain more reticulations than omnians. Now consider a path that contains~$o_1$. Then it is a path that begins in~$U$ and ends in~$R$. Hence, it contains as many omnians as reticulations. It follows that all paths together, i.e. the subgraph of~$B'$ induced by~$S\cup\Gamma(S)$, contains more reticulations than omnians, so $|S|\leq |\Gamma(S)|$. Again, since this holds for all $S\subseteq U'$. It follows from Hall's Theorem that there exists a matching in $B'$ that covers $U'$.\\

iii) Let~$S=U'$. Then~$|S|=|\Gamma(S)| \leo{+ 1}$. Hence, it follows from Hall's Theorem that there does not exist a matching in $B$ that covers $U$.\\

iv) All vertices in $B$ are of degree 2 and it follows in the same way as in \leo{case} i) that $|S|\leq |\Gamma(S)|$ for all~$S\subseteq U'$ and hence that there exists a matching in $B'$ that covers $U'$.\\

Hence, there exists a matching in $B$ that covers $U$ precisely if there is no maximal path that starts and ends in $U$. The theorem now follows from Theorem~\ref{Lstelling}.
\end{proof}

Proposition \ref{Propositie3i)} showed that a binary phylogenetic network is tree-based if for each reticulation both parents are tree-vertices and not tree-based if for at least one reticulation both parents are reticulations.
However, in the situation in which a reticulation in~$N$ has one parent that is a reticulation and the other a tree-vertex it is not immediately clear if~$N$ is tree-based or not. 
The next corollary shows that such networks are tree-based if an additional condition is fulfilled. 

\begin{corollary}
\label{Estelling1}
If for every reticulation~$r$ in a binary phylogenetic network~$N$ either\begin{enumerate}
\item[(i)] both parents of $r$ are tree-vertices; or
\item[(ii)] one parent of $r$ is a tree-vertex and the sibling of $r$ is a tree-vertex or a leaf,
\end{enumerate}
then $N$ is tree-based.
\end{corollary}
\begin{proof} 
Let $B=(U \cup R, E)$ be the bipartite graph associated to $N$. Since all vertices in~$B$ have degree at most two, each connected component of~$B$ is a path or a cycle. Assume that there exists a maximal path~$P$ that starts and ends in~$U$. Let~$u\in U$ be the first vertex on this path. Since~$u$ is an omnian with only one child, it is also a reticulation. Let~$r$ be the only child of~$u$. Since one parent of~$r$ is a reticulation, the other parent~$p$ of~$r$ must be a tree-vertex and the sibling of~$r$ also a tree-vertex or a leaf. However, that means that~$p$ is not an omnian. Hence, path~$P$ ends in~$r$, which is a contradiction to the assumption that~$P$ ended in~$U$. It follows that there is no path that starts and ends in~$U$. By Theorem~\ref{Lstelling} it follows that~$N$ is tree-based.
\end{proof}

The following characterization of binary tree-based phylogenetic networks follows directly from Theorem~\ref{Estelling2}, thus providing an alternative proof of this characterization which was independently discovered (in a slightly different form) by Louxin Zhang~\cite{zhang}. See Figure~\ref{ExamplesALP} for examples. We call a sequence $(u_1,v_1,\ldots ,u_{k},v_{k},u_{k+1})$ of $2k+1$ vertices ($k\geq 1$) of a network~$N$ a \emph{zig-zag} path if~$v_i$ is the child of~$u_i$ and~$u_{i+1}$ for~$i=1,\ldots ,k$.

\begin{corollary}\label{cor:zigzag}
A binary phylogenetic network is tree-based if and only if it contains no zig-zag path $(o_1,r_1,\ldots ,$ $o_{k},r_{k},o_{k+1})$, with $k\geq 1$, in which $r_1,\ldots , r_{k}$ are reticulations, $o_1,\ldots,o_{k+1}$ are omnians and~$o_1$ and~$o_{k+1}$ are reticulations as well as omninans.
\end{corollary}

\section{Nonbinary phylogenetic networks}
\label{sec:nonbin}

\subsection{Preliminaries}
We start with the definition of nonbinary networks.\footnote{Whenever we refer to nonbinary, we mean ``not-necessarily-binary''.} An example is given in Figure~\ref{ExampleNonBinary}.

\begin{definition}
A (\emph{rooted}) \emph{nonbinary phylogenetic network} is a directed, acyclic graph $N=(V,A)$ that \leo{contains a single \emph{root} with indegree~0 and outdegree~1 or more and may additionally }contain the following types of vertices:
\begin{itemize}
\item vertices with indegree $1$ and outdegree $0$, called \emph{leaves} (coloured blue in Figure \ref{ExampleNonBinary}), which are labelled;
\item vertices with outdegree $1$ and indegree $2$ or more, called \emph{reticulations} (marked in pink in Figure \ref{ExampleNonBinary});
\item vertices with indegree $1$ and outdegree $2$ or more, called \emph{tree-vertices}.
\end{itemize}
\end{definition}

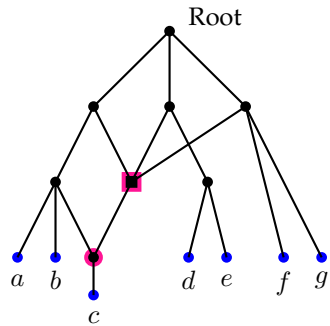
\begin{figure}[h]\centering
\begin{tikzpicture}
\draw (0.6,0.2) node {Root};
\fill (0,0) circle (2pt);
\fill (0,-1) circle (2pt);
\fill (-1,-1) circle (2pt);
\fill (1,-1) circle (2pt);
\fill[color=deeppink] (-1,-3) circle (3.5pt);
\fill (-1.5,-2) circle (2pt);
\filldraw [color=deeppink] ([xshift=-3.5pt,yshift=-3.5pt]-0.5,-2) rectangle ++(7pt,7pt);
\filldraw ([xshift=-2pt,yshift=-2pt]-0.5,-2) rectangle ++(4pt,4pt);
\fill (0.5,-2) circle (2pt);
\fill [blue](1.5,-3) circle (2pt);
\fill [blue](-2,-3) circle (2pt);
\fill (-1,-3) circle (2pt);
\fill [blue](0.25,-3) circle (2pt);
\fill [blue](0.75,-3) circle (2pt);
\fill [blue] (-1,-3.5) circle (2pt);
\fill [blue](-1.5,-3) circle (2pt);
\fill [blue] (2,-3) circle (2pt);

\draw (1.5,-3.35) node {$f$};
\draw (-2,-3.3) node {$a$};
\draw (0.25,-3.3) node {$d$};
\draw (0.75,-3.3) node {$e$};
\draw (-1,-3.8) node {$c$};
\draw (-1.5,-3.3) node {$b$};
\draw (2,-3.3) node {$g$};

\draw [thick] (0,0) -- (0,-1);
\draw [thick] (0,0) -- (1,-1);
\draw [thick] (0,0) -- (-1,-1);
\draw [thick] (-1.5,-2) -- (-1,-1);
\draw [thick] (-1.5,-2) -- (-2,-3);
\draw[thick] (-1,-3) -- (-1.5,-2);
\draw [thick] (0.5,-2) -- (0,-1);
\draw [thick] (-0.5,-2) -- (-1,-1);
\draw [thick] (-0.5,-2) -- (0,-1);
\draw [thick] (-1,-3) -- (-0.5,-2);
\draw [thick] (-0.5,-2) -- (1,-1);
\draw [thick] (0.5,-2) -- (0.25,-3);
\draw [thick] (0.5,-2) -- (0.75,-3);
\draw [thick] (-1,-3) -- (-1,-3.5);
\draw [thick] (1.5,-3) -- (1,-1);
\draw [thick] (1,-1) -- (2,-3);
\draw [thick] (-1.5,-2) -- (-1.5,-3);

\end{tikzpicture}
\caption{Example of a nonbinary phylogenetic network.}
\label{ExampleNonBinary}
\end{figure}

A \emph{nonbinary phylogenetic tree} is a nonbinary phylogenetic network without reticulations.

%
%
%

We will consider two different variants of tree-basedness of nonbinary networks, which we name ``tree-based'' and ``strictly-tree-based''.

\begin{definition} A nonbinary phylogenetic network $N$ is called \emph{tree-based} with base-tree $T$, when $N$ can be obtained from $T$ via the following steps:
\begin{enumerate}
\item[(a)] Add some vertices to arcs of $T$. These vertices, called \emph{attachment points}, have in- and outdegree $1$.
\item[(b)] Add arcs, called \emph{linking arcs}, between pairs of attachments points and from tree-vertices to attachment points, so that $N$ remains acyclic and so that attachment points have indegree or outdegree $1$.
\item[(c)] Suppress every attachment point that is not incident to a linking arc.
\end{enumerate}
\end{definition}

\begin{definition} A nonbinary phylogenetic network $N$ is called \emph{strictly tree-based} with base-tree $T$, when $N$ can be obtained from $T$ via the following steps:
\begin{enumerate}
\item[(1)] Add some vertices to arcs in $T$. These vertices, called attachment points, have in- and outdegree $1$.
\item[(2)] Add arcs, called \emph{linking arcs}, between pairs of attachments points, so that $N$ remains acyclic and so that exactly one linking arc is attached to each attachment point.
\end{enumerate}
\end{definition}

A nonbinary phylogenetic network is \emph{tree-based} if it is tree-based with base-tree~$T$ for some nonbinary phylogenetic tree~$T$. Similarly, a nonbinary phylogenetic network is \emph{strictly tree-based} if it is strictly tree-based with base-tree~$T$ for some nonbinary rooted phylogenetic tree~$T$.

The distinction between tree-based and strictly tree-based is illustrated by two examples in Figure~\ref{fig:strictly}. An example of a strictly-tree-based nonbinary network can be found in Figure~\ref{fig:strictlytreebased}(a).

\begin{figure}[h]
    \centering
    \subfigure[This network is not strictly tree-based because the linking arc is attached to a vertex of the base-tree, rather than to an attachment point, and this cannot be avoided.]{\begin{tikzpicture}
\draw [thick] (0,0)  -- (-1,-1) ; 
\draw [thick] (0,0)  -- (1,-1) ; 
\draw [thick] (1,-1)  -- (1.2,-2.5) ; 
\draw [thick] (1,-1)  -- (1.7,-2.5) ; 
\draw [thick, dashed, gray] (1,-1)  -- (0,-2) ; 
\draw [thick] (-1,-1)  -- (-1.2,-2.5) ; 
\draw [thick] (-1,-1)  -- (-1.7,-2.5) ; 
\draw [thick] (-1,-1)  -- (0,-2) ; 
\draw [thick] (0,-2) -- (0,-2.5);

\fill (0,0) circle (2pt); 
\fill (-1,-1) circle (2pt); 
\fill (1,-1) circle (2pt); 
\fill (-1.7,-2.5) circle (2pt);
\fill (-1.2,-2.5) circle (2pt); 
\fill (0,-2.5) circle (2pt); 
\fill (1.2,-2.5) circle (2pt); 
\fill (1.7,-2.5) circle (2pt); 

\fill[color=deeppink] (0,-2) circle (3.5pt); 
\fill (0,-2) circle (2pt);

\draw (-1.7,-2.75) node {$a$};      
\draw (-1.2,-2.75) node {$b$};      
\draw (0,-2.75) node {$c$};      
\draw (1.7,-2.75) node {$e$};      
\draw (1.2,-2.75) node {$d$};      
\end{tikzpicture}}
\hspace{2cm}
    \subfigure[This network is not strictly tree-based because two linking arcs are attached to the same attachment point, and this cannot be avoided.]{\begin{tikzpicture}

\draw [thick] (0,0)  -- (-1,-1) ; 
\draw [thick] (0,0)  -- (1,-1) ; 
\draw [thick] (0,0)  -- (0,-1) ;
\draw [thick, dashed, gray] (0,-1)  -- (-0.5,-2) ; 
\draw [thick] (0,-1)  -- (0.5,-2) ; 
\draw [thick, gray, dashed] (1,-1)  -- (-0.5,-2) ; 
\draw [thick, gray, dashed] (1,-1)  -- (0.5,-2) ; 
\draw [thick] (1,-1)  -- (1.7,-2.5) ; 
\draw [thick] (-1,-1)  -- (-1.7,-2.5) ; 
\draw [thick] (-1,-1)  -- (-0.5,-2) ; 
\draw [thick] (-0.5,-2) -- (-0.5,-2.5);
\draw [thick] (0.5,-2) -- (0.5,-2.5);

\fill (0,0) circle (2pt); 
\fill (-1,-1) circle (2pt); 
\fill (1,-1) circle (2pt);
\filldraw ([xshift=-2pt,yshift=-2pt]0,-1) rectangle ++(4pt,4pt);
\fill (-1.7,-2.5) circle (2pt);
\fill (-0.5,-2.5) circle (2pt); 
\fill (0.5,-2.5) circle (2pt);
\fill (1.7,-2.5) circle (2pt); 
 
\fill[color=deeppink] (-0.5,-2) circle (3.5pt); 
\fill (-0.5,-2) circle (2pt);
\fill[color=deeppink] (0.5,-2) circle (3.5pt); 
\fill (0.5,-2) circle (2pt);

\draw (-1.7,-2.75) node {$a$};      
\draw (-0.5,-2.75) node {$b$};      
\draw (0.5,-2.75) node {$c$};      
\draw (1.7,-2.75) node {$d$};      
\end{tikzpicture}}
\caption{Two nonbinary networks that are both tree-based but not strictly tree-based. The black solid lines indicate possible base-trees, while the dashed, grey lines are the linking arcs.\label{fig:strictly}}
\end{figure}
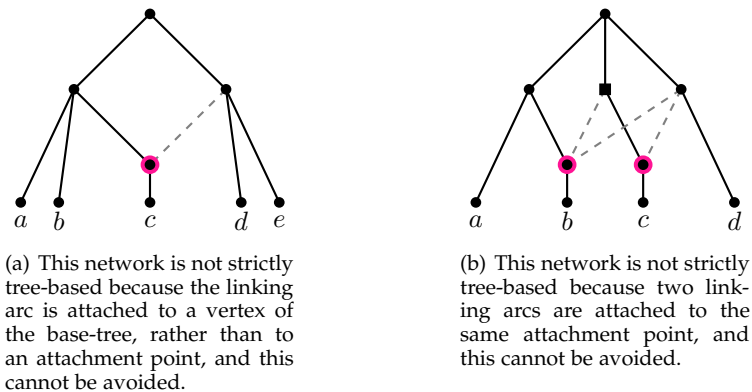

The following observations are easily verified.

\begin{observation}
Let~$N$ be a nonbinary phylogenetic network. If~$N$ is strictly tree-based, then~$N$ is tree-based.
\end{observation}

If~$N$ and~$N'$ are nonbinary phylogenetic networks, then we say that~$N'$ is a \emph{refinement} of~$N$ if~$N$ can be obtained from~$N'$ by contracting some of its edges.

\begin{observation}
Let~$N$ be a nonbinary phylogenetic network. Then~$N$ is tree-based if and only if there exists a binary refinement of~$N$ that is tree-based.
\end{observation}


Any definitions from Section \ref{preliminaries} that have not been mentioned in this section, are defined similarly as in the binary case.

We first discuss nonbinary tree-based networks in Section~\ref{sec:nonbintreebased} and then nonbinary strictly-tree-based networks in Section~\ref{sec:strictly}.

\subsection{Nonbinary tree-based phylogenetic networks}\label{sec:nonbintreebased}
We will examine if some of the theorems from Section~\ref{subsec:binresults} hold for nonbinary phylogenetic networks as well. First, we look at te stability of networks.

\begin{proposition}
A nonbinary stable network $N$ has the following property:
\leo{the} child and the parents of every reticulation are tree-vertices.
\end{proposition}
\begin{proof}
Can be shown similar to the proof of Proposition~\ref{4.1}~\cite{Artikel3}.
\end{proof}

Hence, this property holds also for nonbinary networks. Next, we will consider the following two questions.
\begin{enumerate}
\item[(i)] Is every nonbinary stable phylogenetic network tree-based? (Corollary \ref{Corollary} in the binary case)
\item[(ii)] For a nonbinary phylogenetic network $N$, is $N$ tree-based if all parents of all reticulations of $N$ are tree-vertices? (Proposition \ref{Propositie3i)} in the binary case)
\end{enumerate}

There is one single example that answers both of the questions. The example, displayed in Figure~\ref{Counterexample}, shows that the answer to both questions is ``no''. These properties only hold in the binary case.

\begin{figure}[h]\centering
\begin{tikzpicture}
\fill (0,0) circle (2pt);
\filldraw ([xshift=-2pt,yshift=-2pt]-1,-1) rectangle ++(4pt,4pt);
\filldraw ([xshift=-2pt,yshift=-2pt]0,-1) rectangle ++(4pt,4pt);
\filldraw ([xshift=-2pt,yshift=-2pt]1,-1) rectangle ++(4pt,4pt);
\fill [deeppink] (-0.5,-2) circle (3.5pt);
\fill [deeppink] (0.5,-2) circle (3.5pt);
\fill (-0.5,-2) circle (2pt);
\fill (0.5,-2) circle (2pt);
\fill (0.5,-2.8) circle (2pt);
\fill (-0.5,-2.8) circle (2pt);
\draw [thick] (0,0) -- (-1,-1);
\draw [thick] (0,0) -- (0,-1);
\draw [thick] (0,0) -- (1,-1);
\draw [thick] (1,-1) -- (0.5,-2);
\draw [thick] (1,-1) -- (-0.5,-2);
\draw [thick] (0,-1) -- (0.5,-2);
\draw [thick] (0,-1) -- (-0.5,-2);
\draw [thick] (-1,-1) -- (0.5,-2);
\draw [thick] (-1,-1) -- (-0.5,-2);
\draw [thick] (-0.5,-2) -- (-0.5,-2.8);
\draw [thick] (0.5,-2) -- (0.5, -2.8);
\end{tikzpicture}
\caption{A nonbinary network that shows that the answer to Questions~(i) and~(ii) is negative.}
\label{Counterexample}
\end{figure}
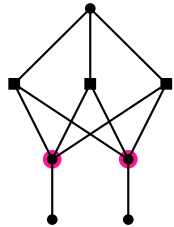

We will now show that Theorem~\ref{Lstelling} even holds in the nonbinary case.

\begin{theorem}
\label{LstellingNB}
Given a nonbinary phylogenetic network $N$ and the bipartite graph~$B$ that is associated to $N$. Network~$N$ is tree-based if and only if there exists a matching $M$ in~$B$ with~$|M|=|U|$.
\end{theorem}
\begin{proof} If there exists a matching $M$ in $B$ with~$|M|=|U|$, then it can be proved similarly as in the binary case (Theorem~\ref{Lstelling}), that $N$ is tree-based.

Now assume that $N$ is tree-based. Then it can be proved partially similar as in the binary case, that there exists a matching in~$B$ that covers all omnians. 
The only difference is that when an omnian has more than one outgoing arc contained in a base-tree~$T$, that only one edge should be coloured and the rest should not be coloured in $B$. The rest of the proof is the same as in the proof of Theorem~\ref{Lstelling}.
\end{proof}

This theorem directly leads to a polynomial-time algorithm for deciding if a network is tree-based, using one of the algorithms for maximum cardinality bipartite matching (see e.g.~\cite{schrijver}).

\begin{corollary}
There exists a polynomial-time algorithm that decides whether a given nonbinary phylogenetic network is tree-based or not.
\end{corollary}

Consider Hall's Theorem (Theorem \ref{Hall}) and Theorem \ref{LstellingNB}. Combining those two theorems gives a characterization for nonbinary tree-based phylogenetic networks, similar to Corollary~\ref{Aantalkind} in the binary case.

\begin{corollary}\label{AantalkindNB}
Let $N$ be a nonbinary phylogenetic network and $U$ the set of all omnians of~$N$. 
Then~$N$ is tree-based if and only if  for all $S \subseteq U$ the number of different children of \leo{the vertices in}~$S$ is greater than or equal to the number of omnians in~$S$.
\end{corollary}

In Theorem~\ref{Estelling2} we showed that a binary network is tree-based if and only if the associated bipartite graph contains no maximal path which starts and ends in $U$. One might suspect that this also holds in the nonbinary case. We will look at a partial nonbinary phylogenetic network $N$, which is displayed in Figure~\ref{Nonbinaryexample}(a), 
and the bipartite graph $B$ that is associated to $N$, which is displayed in~\ref{Nonbinaryexample}(b). A matching is drawn in $B$, which is coloured blue and \leo{dash-dotted} in Figure \ref{Nonbinaryexample}(b). 
We see that in $B$ there is a maximal path starting and ending in $U$: starting in $b$ via $f - c - g$ ending in $d$.
Though in the binary case this would mean that $N$ is not tree-based, we see in Figure \ref{Nonbinaryexample}(b) that there exists a matching that covers $U$. With Theorem \ref{LstellingNB} it follows that $N$ is tree-based.

Therefore, for a nonbinary phylogenetic network~$N$ and the bipartite graph $B=(U \cup R,E)$ associated to $N$, if there is a maximal path starting and ending in $U$, then $N$ can still be tree-based.

Consequently, also Corollary~\ref{cor:zigzag} does not hold in the nonbinary case.

\begin{flushleft}
\begin{figure}[h]\centering
\subfigure[The \leo{dash-dotted} lines represent the matching in (b).]{\begin{tikzpicture}
\draw [thick] (1,1) -- (1,0) -- (1.5,-1);
\draw [thick, dashdotted, blue] (1.5,-1) -- (3,0);
\draw [thick, dashdotted, blue] (2,0) -- (2.5,-1);
\draw [thick] (2.5,-1) -- (3,0);
\draw [thick] (1.35,1) -- (2,0);
\draw [thick] (2.65,1) -- (2,0);
\draw [thick] (3,0) -- (3.5,-1);
\draw [thick, dashdotted, blue] (3.5,-1) -- (4,0);
\draw [thick] (3.35,1) -- (4,0);
\draw [thick] (4.65,1) -- (4,0);
\draw [thick] (2.5,-1) -- (2.5,-1.5);
\draw [thick]  (3.5,-1) -- (3.5,-1.5);
\draw [thick] (1.5,-1) -- (1.5,-1.5);
\draw [thick] (3,0) -- (3,1);
\draw [thick] (1,0) -- (0.5,-1);
\fill (1,0) circle (2pt);
\fill [deeppink] (1.5,-1) circle (3.5pt);
\fill (1.5,-1) circle (2pt);
\filldraw [color=deeppink] ([xshift=-3.5pt,yshift=-3.5pt]2,0) rectangle ++(7pt,7pt);
\filldraw ([xshift=-2pt,yshift=-2pt]2,0) rectangle ++(4pt,4pt);
\filldraw ([xshift=-2pt,yshift=-2pt]3,0) rectangle ++(4pt,4pt);
\fill [deeppink] (2.5,-1) circle (3.5pt);
\fill (2.5,-1) circle (2pt);
\filldraw [color=deeppink] ([xshift=-3.5pt,yshift=-3.5pt]4,0) rectangle ++(7pt,7pt);
\filldraw ([xshift=-2pt,yshift=-2pt]4,0) rectangle ++(4pt,4pt);
\fill [deeppink] (3.5,-1) circle (3.5pt);
\fill (3.5,-1) circle (2pt);
\draw (0.7,0) node {$a$};
\draw (1.7,0) node {$b$};
\draw (2.7,0) node {$c$};
\draw (4.3,0) node {$d$};
\draw (1.2,-1) node {$e$};
\draw (2.7,-1) node {$f$};
\draw (3.2,-1) node {$g$};
\end{tikzpicture}}
\hspace{2cm}
\subfigure[A matching is $B$ indicated in \leo{dash-dotted} lines.]{\begin{tikzpicture}
\draw [thick, blue, dashdotted] (0,0) -- (2.5,-1.8);
\draw [thick, blue, dashdotted](0,-1) -- (2.5,-1);
\draw [thick](0,-1) -- (2.5,-1.8);
\draw [thick](0,-1) -- (2.5,-2.6);
\draw [thick, blue, dashdotted](0,-2) -- (2.5,-2.6);
\filldraw [color=deeppink] ([xshift=-3.5pt,yshift=-3.5pt]0,0) rectangle ++(7pt,7pt);
\filldraw ([xshift=-2pt,yshift=-2pt]0,0) rectangle ++(4pt,4pt);
\filldraw ([xshift=-2pt,yshift=-2pt]0,-1) rectangle ++(4pt,4pt);
\filldraw [color=deeppink] ([xshift=-3.5pt,yshift=-3.5pt]0,-2) rectangle ++(7pt,7pt);
\filldraw ([xshift=-2pt,yshift=-2pt]0,-2) rectangle ++(4pt,4pt);
\filldraw [color=deeppink] ([xshift=-3.5pt,yshift=-3.5pt]2.5,0.6) rectangle ++(7pt,7pt);
\filldraw ([xshift=-2pt,yshift=-2pt]2.5,0.6) rectangle ++(4pt,4pt);
\filldraw [color=deeppink] ([xshift=-3.5pt,yshift=-3.5pt]2.5,-0.2) rectangle ++(7pt,7pt);
\filldraw ([xshift=-2pt,yshift=-2pt]2.5,-0.2) rectangle ++(4pt,4pt);
\fill[deeppink] (2.5,-1) circle (3.5pt);
\fill[deeppink] (2.5,-1.8) circle (3.5pt);
\fill[deeppink] (2.5,-2.6) circle (3.5pt);
\fill (2.5,-1) circle (2pt);
\fill (2.5,-1.8) circle (2pt);
\fill (2.5,-2.6) circle (2pt);
\draw (-0.3,0) node {$b$};
\draw (-0.3,-1) node {$c$};
\draw (-0.3,-2) node {$d$};
\draw (2.8,0.6) node {$b$};
\draw (2.8,-0.2) node {$d$};
\draw (2.8,-1) node {$e$};
\draw (2.8,-1.8) node {$f$};
\draw (2.8,-2.6) node {$g$};
\draw (0,1.1) node {$U$};
\draw (2.5,1.1) node {$R$};
\end{tikzpicture}}
\caption{A partial nonbinary phylogenetic network and the bipartite graph $B$ that is associated to $N$, showing that Theorem~\ref{Estelling2} does not hold in the nonbinary case.}
\label{Nonbinaryexample}
\end{figure}
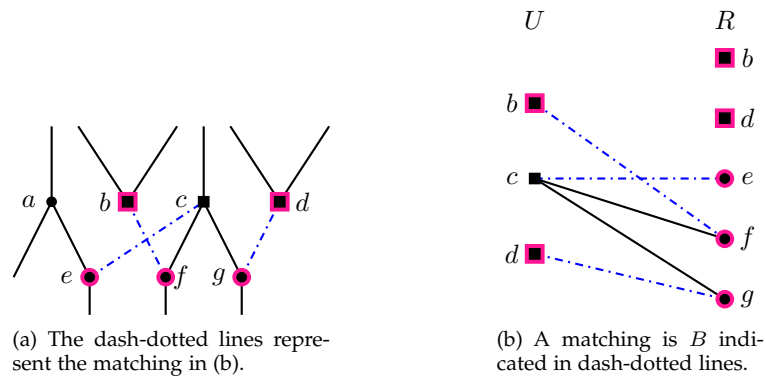
\end{flushleft}

%
%

\subsection{Nonbinary strictly-tree-based phylogenetic networks}\label{sec:strictly}

Here we show that the characterization of binary tree-based phylogenetic networks in Corollary~\ref{cor:zigzag} can be extended to a characterization for nonbinary strictly-tree-based phylogenetic networks. We call vertices with outdegree greater than two \emph{multifurcations}.

In this case, we use a modified bipartite graph. Let $N=(V,A)$ be a nonbinary phylogenetic network. The \emph{modified bipartite graph associated to} $N$ is the bipartite graph $B^*=(U \cup R,E)$, which is defined as follows. \leo{For each vertex~$v\in V$ of~$N$ that is an omnian with outdegree~2, we put a vertex~$v_o$ in~$U$. For each vertex~$w\in V$ of~$N$ that is a reticulation, we put a vertex~$w_r$ in~$R$. We put an edge~$\{v_o,w_r\}$ in~$E$ for each~$v_o\in U$ and~$w_r\in R$ with $(v,w) \in A$. Then, for each multifurcation~$u\in V$ of~$N$ that has~$k$ children that are reticulations $w_1,\ldots ,w_k$, with~$k\geq 1$, we add~$k$ vertices~$u_1,\ldots ,u_k$ to~$U$ and add edges \rev{$\{u_i,{w_i}_r\}$} for~$i=1,\ldots ,k$ to~$E$ \rev{(with~${w_i}_r$ the reticulation in~$R$ corresponding to reticulation~$w_i$ in~$V$.)}. As with the previous bipartite graphs, we will omit the subscripts of the vertex labels and refer to~$v_o,w_r,u_i$ simply as~$v,w$ and~$u$, respectively.} Examples are given in Figure~\ref{fig:strictlytreebased2} and~\ref{fig:strictlytreebased}. \leo{Note in particular that in~$B^*$ the set~$U$ contains not only omnians but also multifurcations that may have non-reticulate children.}

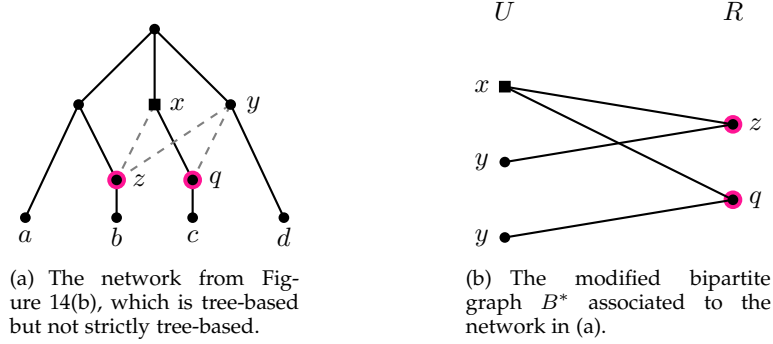
\begin{figure}[h]
    \centering
    \subfigure[The network from Figure~\ref{fig:strictly}(b), which is tree-based but not strictly tree-based.]{\begin{tikzpicture}
\draw [thick] (0,0)  -- (-1,-1) ; 
\draw [thick] (0,0)  -- (1,-1) ; 
\draw [thick] (0,0)  -- (0,-1) ;
\draw [thick, dashed, gray] (0,-1)  -- (-0.5,-2) ; 
\draw [thick] (0,-1)  -- (0.5,-2) ; 
\draw [thick, gray, dashed] (1,-1)  -- (-0.5,-2) ; 
\draw [thick, gray, dashed] (1,-1)  -- (0.5,-2) ; 
\draw [thick] (1,-1)  -- (1.7,-2.5) ; 
\draw [thick] (-1,-1)  -- (-1.7,-2.5) ; 
\draw [thick] (-1,-1)  -- (-0.5,-2) ; 
\draw [thick] (-0.5,-2) -- (-0.5,-2.5);
\draw [thick] (0.5,-2) -- (0.5,-2.5);

\fill (0,0) circle (2pt); 
\fill (-1,-1) circle (2pt); 
\fill (1,-1) circle (2pt);
\filldraw ([xshift=-2pt,yshift=-2pt]0,-1) rectangle ++(4pt,4pt);
\fill (0,-1) circle (2pt);
\fill (-1.7,-2.5) circle (2pt);
\fill (-0.5,-2.5) circle (2pt); 
\fill (0.5,-2.5) circle (2pt);
\fill (1.7,-2.5) circle (2pt); 
 
\fill[color=deeppink] (-0.5,-2) circle (3.5pt); 
\fill (-0.5,-2) circle (2pt);
\fill[color=deeppink] (0.5,-2) circle (3.5pt); 
\fill (0.5,-2) circle (2pt);

\draw (-1.7,-2.75) node {$a$};      
\draw (-0.5,-2.75) node {$b$};      
\draw (0.5,-2.75) node {$c$};      
\draw (1.7,-2.75) node {$d$};      
\draw (0.3,-1) node {$x$}; 
\draw (1.3,-1) node {$y$}; 
\draw (-0.2,-2) node {$z$};
\draw (0.8,-2) node {$q$};  
\end{tikzpicture}}
\hspace{2cm}
\subfigure[The modified bipartite graph~$B^*$ associated to the network in~(a).]{\begin{tikzpicture}
\filldraw ([xshift=-2pt,yshift=-2pt]0,-1) rectangle ++(4pt,4pt);
\fill (0,-2) circle (2pt);
\fill (0,-3) circle (2pt);
\fill[color=deeppink] (3,-1.5) circle (3.5pt); 
\fill[color=deeppink] (3,-2.5) circle (3.5pt); 
\fill (3,-1.5) circle (2pt);
\fill (3,-2.5) circle (2pt);

\draw (0,0) node {$U$};
\draw (3,0) node {$R$};
\draw (-0.3,-1) node {$x$};
\draw (-0.3,-2) node {$y$};
\draw (-0.3,-3) node {$y$};
\draw (3.3,-1.5) node {$z$};
\draw (3.3,-2.5) node {$q$};

\draw [thick] (0,-1) -- (3,-1.5);
\draw [thick] (0,-1) -- (3,-2.5);
\draw [thick] (0,-2) -- (3,-1.5);
\draw [thick] (0,-3) -- (3,-2.5);
\end{tikzpicture}}
\caption{\label{fig:strictlytreebased2} It is easy to see that the modified bipartite graph~$B^*$ associated to the nonbinary network in~(a) has no matching that covers~$U$. Hence, the network in~(a) is not strictly tree-based.}
\end{figure}

\begin{figure}[h]\centering
\subfigure[A strictly-tree-based nonbinary network.]{
\begin{tikzpicture}
\draw [thick] (0,0) -- (1,-1);
\draw [thick] (0,0) -- (-1,-1);
\draw [thick] (-1,-1) -- (-0.7,-1.5);
\draw [thick] (-1,-1) -- (-1.7,-2.5);
\draw [thick] (0,0) -- (0,-1);
\draw [thick] (0.7,-1.5) -- (1,-1);
\draw [thick] (1.7,-2.5) -- (1,-1);
\draw [thick, gray, dashed] (-0.7,-1.5) -- (-0.35,-2);
\draw [thick] (0,-1) -- (0.35,-2);
\draw [thick] (0,-1) -- (0,-2.5);
\draw [thick] (0,-1) -- (-0.35,-2);
\draw [thick, gray, dashed] (0.7,-1.5) -- (0.35,-2);
\draw [thick] (-0.35,-2.5) -- (-0.35,-2);
\draw [thick] (0.35,-2.5) -- (0.35,-2);

\fill [deeppink] (0.35,-2) circle (3.5pt); 
\fill [deeppink] (-0.35,-2) circle (3.5pt);
\fill (0,0) circle (2pt); 
\fill (1,-1) circle (2pt);
\fill (-1,-1) circle (2pt);
\fill (0.7,-1.5) circle (2pt);
\fill (1.7,-2.5) circle (2pt);
\fill (0.35,-2.5) circle (2pt);
\fill (1,-2.5) circle (2pt);
\fill (0,-2.5) circle (2pt);
\fill (0,-1) circle (2pt); 
\fill (-0.7,-1.5) circle (2pt);
\fill (-1.7,-2.5) circle (2pt);
\fill (-0.35,-2.5) circle (2pt);
\fill (-1,-2.5) circle (2pt);
\fill (0.35,-2) circle (2pt);
\fill (-0.35,-2) circle (2pt);
\draw [thick] (0.7,-1.5) -- (1,-2.5);
\draw [thick] (-0.7,-1.5) -- (-1,-2.5);
\draw (-1.7,-2.8) node {$a$};
\draw (-1,-2.8) node {$b$};
\draw (-0.35,-2.8) node {$c$};
\draw (0,-2.8) node {$d$};
\draw (0.35,-2.8) node {$e$};
\draw (1,-2.8) node {$f$};
\draw (1.7,-2.8) node {$g$};
\draw (0.3,-1) node {$x$};
\draw (-0.55,-2.2) node {$y$};
\draw (0.50,-2.2) node {$z$};
\end{tikzpicture}}\hspace{2cm}
\subfigure[The modified bipartite graph~$B^*$ associated to the network in~(a).]{\begin{tikzpicture}
\fill (0,-1) circle (2pt);
\fill (0,-2) circle (2pt);
\fill [deeppink] (3,-1) circle (3.5pt);
\fill [deeppink] (3,-2) circle (3.5pt);
\fill (3,-1) circle (2pt);
\fill (3,-2) circle (2pt);

\draw (0,0) node {$U$};
\draw (3,0) node {$R$};
\draw (-0.3,-1) node {$x$};
\draw (-0.3,-2) node {$x$};
\draw (3.3,-1) node {$y$};
\draw (3.3,-2) node {$z$};

\draw [thick] (0,-1) -- (3,-1);
\draw [thick] (0,-2) -- (3,-2);
\end{tikzpicture}}
\caption{\label{fig:strictlytreebased} It is easy to see that the modified bipartite graph~$B^*$ associated to the nonbinary network in~(a) has a matching that covers~$U$. Hence, the network in~(a) is strictly tree-based.}
\end{figure}
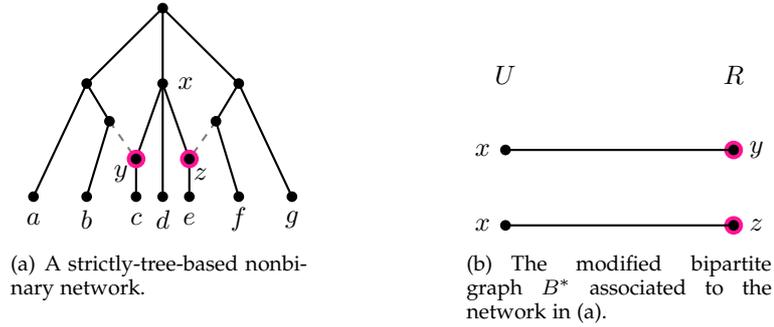

Using this modified bipartite graph, we can proceed in the same way as in the binary case.

\begin{theorem}
\label{thm:strictly}
Let $N$ be a nonbinary phylogenetic network and $B^*~=~(U~\cup~R,E)$ the modified bipartite graph associated to~$N$. Network~$N$ is strictly tree-based if and only if all reticulations of~$N$ have indegree~2 and there exists a matching~$M$~in~$B^*$ so that $\left| U\right| = \left| M \right|$.
\end{theorem}
\begin{proof}
It is clear that~$N$ cannot be strictly tree-based if at least one of the reticulations has indegree greater than two, because no two linking arcs are allowed to attach to the same attachment point. Hence, we may assume from now on that all reticulations have indegree~2.

First assume that~$N$ is strictly tree-based with base-tree~$T$. Consider the set of edges~$E'$ of~$B^*$ that correspond to arcs of~$N$ that are contained in~$T$ (i.e., that are not linking arcs w.r.t. base-tree~$T$). For each multifurcation, each outgoing arc must be contained in~$T$. Moreover, as in the binary case, for each omnian, at least one outgoing arc must be contained in~$T$. Hence, the set~$E'$ of edges touches all vertices in~$U$. Each vertex~$r\in R$ has exactly one incident edge in~$E'$ because~$T$ is a tree. For each vertex~$u\in U$ that has two incident edges in~$E'$, remove one of them, arbitrarily. This gives a new set of edges~$M\subseteq E'$, which is a matching in~$B^*$ with~$|M|=|U|$.

Now assume that there exists such a matching~$M$. As in the binary case, we construct a set~$A$ of arcs by adding the outgoing arc of every reticulation, the incoming arc of every tree-vertex, every arc corresponding to an edge in~$M$ (if it has not yet been added) and for every reticulation that has not yet been covered, one of its incomming arcs, arbitrarily. Consider the tree $T$ consisting of all vertices of $N$ and the set of arcs $A$. As in the binary case, there are no dummy leaves in~$T$ because matching~$M$ covers all omnians of~$N$. Moreover, for all multifurcations of~$N$, all outgoing arcs are in~$T$ because matching~$M$ contains all corresponding edges of~$B^*$ (since they are incident to a degree-1 vertex). Hence, each arc of~$N$ that is not in~$A$ connects an outdegree-2 vertex with an indegree-2 vertex. These are the linking arcs, and their endpoints the attachment points. Hence, each linking arc is attached to two attachment points, and no two linking arcs are attached to the same attachment point. Hence,~$N$ is strictly tree-based.
\end{proof}

From the above theorem it follows directly that it can be decided in polynomial time whether a nonbinary network is strictly tree-based, using one of the algorithms for maximum cardinality bipartite matching (see e.g.~\cite{schrijver}).

\begin{theorem}\label{thm:strictly2}
Let $N$ be a nonbinary phylogenetic network and $B^*~=~(U~\cup~R,E)$ the modified bipartite graph associated to~$N$. Network~$N$ is strictly tree-based if and only if all reticulations in~$N$ have indegree~2 and~$B^*$ contains no maximal path which starts and ends in $U$.
\end{theorem}
\begin{proof}
It is again clear that~$N$ cannot be strictly tree-based if at least one of the reticulations has indegree greater than two. If all reticulations have indegree 2, then all vertices~$r\in R$ have degree at most two in~$B^*$. Moreover, all vertices~$u\in U$ have degree at most two in~$B^*$ because each multifurcation of~$N$ has been split into multiple vertices in~$U$ with one incident edge each. Hence, bipartite graph~$B^*$ has maximum degree~2 and we can proceed as in the proof of Theorem~\ref{Estelling2}.
\end{proof}

The following characterization in terms of zig-zag paths follows directly from the theorem above. Recall that we call a sequence $(u_1,v_1,\ldots ,u_{k},v_{k},u_{k+1})$ of $2k+1$ vertices ($k\geq 1$) of a network~$N$ a \emph{zig-zag} path if~$v_i$ is the child of~$u_i$ and~$u_{i+1}$ for~$i=1,\ldots ,k$.

\begin{corollary}\label{cor:zigzag}
A nonbinary phylogenetic network~$N$ is strictly tree-based if and only if every reticulation has indegree~2 and there is no zig-zag path $(s,r_1,o_1,r_2,\ldots ,o_{k-1},r_{k},t)$, with~$k\geq 1$, in which $r_1,\ldots , r_{k}$ are reticulations, $o_1,\ldots,o_{k-1}$ are omnians and each of~$s$ and~$t$ is either a multifurcation or a reticulation and an omnian.
\end{corollary}

\section{\rev{Application to biological phylogenetic networks}}\label{sec:example}

\rev{To show how our theorems can be applied to real networks, we discuss two examples of biological phylogenetic networks in this section. Since both networks are nonbinary, previously known theorems do not apply to them. In the following two subsections, we will show how one can use the theorems from this paper to determine whether each of these networks is (strictly) tree-based or not.}

\subsection{Viola network}

\rev{A nonbinary phylogenetic network for violets from the \emph{Viola} genus, based on the network published in~\cite{Violets2}, is displayed} in Figure~\ref{fig:violets}. The \emph{Viola} genus contains about~600 species, which are divided over sixteen different sections. The leaves of the network in Figure~\ref{fig:violets} represent these sections, relabelled as follows: $a=$~\emph{Sect.~nov.~A}, $b=$~\emph{Sect.~nov.~B}, $c=$~\emph{Melanium}, $d=$~\emph{Delphiniopsis}, $e=$~\emph{Sclerosium}, $f=$~\emph{Viola~s.str.}, $g=$~\emph{Plagiostigma}, $h=$~\emph{Nosphinium~s.lat.}, $i=$~\emph{Xylinosium}, $j=$~\emph{Chamaemelanium}, $k=$~\emph{Chilenium}, $l=$~\emph{Erpetion}, $m=$~\emph{Rubellium}, $n=$~\emph{Tridens}, $o=$~\emph{Leptidium}, $p=$~\emph{Andinium}. The~21 reticulations in the network (indicated with pink shading around the nodes) represent polyploidisations.

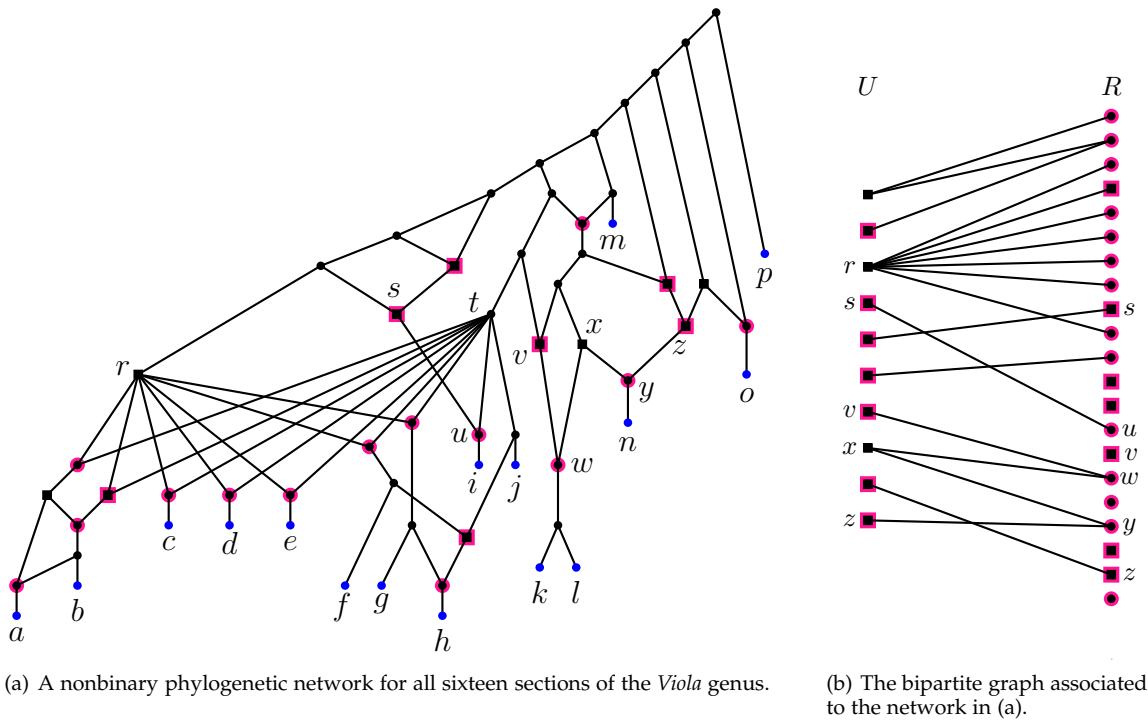
\begin{figure}[h]
\centering
\subfigure[A nonbinary phylogenetic network for all sixteen sections of the \emph{Viola} genus.]{
\begin{tikzpicture}[scale=.8]
\fill[color=deeppink] (0,0.5) circle (3.5pt); 
\fill[color=deeppink] (1,1.5) circle (3.5pt); 
\fill[color=deeppink] (2.5,2) circle (3.5pt); 
\fill[color=deeppink] (3.5,2) circle (3.5pt); 
\fill[color=deeppink] (4.5,2) circle (3.5pt); 
\fill[color=deeppink] (5.8,2.8) circle (3.5pt); 
\fill[color=deeppink] (6.5,3.2) circle (3.5pt); 
\fill[color=deeppink] (1,2.5) circle (3.5pt); 
\fill[color=deeppink] (7,0.5) circle (3.5pt); 
\fill[color=deeppink] (7.6,3) circle (3.5pt); 
\fill[color=deeppink] (8.9,2.5) circle (3.5pt); 
\fill[color=deeppink] (10.7,5.5) circle (3.5pt); 
\fill[color=deeppink] (12,4.8) circle (3.5pt); 
\fill[color=deeppink] (10.05,3.9) circle (3.5pt); 
\fill[color=deeppink] (9.3,6.5) circle (3.5pt); 
\fill (0,0.5) circle (2pt); 
\filldraw ([xshift=-2pt,yshift=-2pt]0.5,2) rectangle ++(4pt,4pt); 
\fill (1,1) circle (2pt); 
\fill (1,1.5) circle (2pt); 
\filldraw [deeppink] ([xshift=-3.5pt,yshift=-3.5pt]1.5,2) rectangle ++(7pt,7pt); 
\filldraw ([xshift=-2pt,yshift=-2pt]1.5,2) rectangle ++(4pt,4pt); 
\fill (2.5,2) circle (2pt); 
\fill (3.5,2) circle (2pt); 
\fill (4.5,2) circle (2pt); 
\fill (5.8,2.8) circle (2pt); 
\fill (6.5,3.2) circle (2pt); 
\fill (1,2.5) circle (2pt); 
\filldraw ([xshift=-2pt,yshift=-2pt]2,4) rectangle ++(4pt,4pt);
\fill (6.2,2.2) circle (2pt); 
\fill (6.5,1.5) circle (2pt); 
\filldraw [deeppink] ([xshift=-3.5pt,yshift=-3.5pt]7.4,1.3) rectangle ++(7pt,7pt);
\filldraw ([xshift=-2pt,yshift=-2pt]7.4,1.3) rectangle ++(4pt,4pt);
\fill (7,0.5) circle (2pt); 
\fill (7.6,3) circle (2pt); 
\fill (8.2,3) circle (2pt); 
\filldraw [deeppink] ([xshift=-3.5pt,yshift=-3.5pt]6.25,5) rectangle ++(7pt,7pt);
\filldraw ([xshift=-2pt,yshift=-2pt]6.25,5) rectangle ++(4pt,4pt);
\fill (5,5.8) circle (2pt); 
\filldraw [deeppink] ([xshift=-3.5pt,yshift=-3.5pt]7.2,5.8) rectangle ++(7pt,7pt);
\filldraw ([xshift=-2pt,yshift=-2pt]7.2,5.8) rectangle ++(4pt,4pt);
\fill (6.25,6.3) circle (2pt); 
\fill (7.8,7) circle (2pt); 
\fill (8.3,6) circle (2pt); 
\filldraw [deeppink] ([xshift=-3.5pt,yshift=-3.5pt]8.6,4.5) rectangle ++(7pt,7pt);
\filldraw ([xshift=-2pt,yshift=-2pt]8.6,4.5) rectangle ++(4pt,4pt);
\fill (8.8,7) circle (2pt); 
\fill (8.6,7.5) circle (2pt); 
\fill (9.5,8) circle (2pt); 
\fill (9.8,7) circle (2pt); 
\fill (9.3,6.5) circle (2pt); 
\fill (9.3,6) circle (2pt); 
\fill (8.9,5.5) circle (2pt); 
\fill (8.9,2.5) circle (2pt); 
\fill (8.9,1.5) circle (2pt); 
\filldraw ([xshift=-2pt,yshift=-2pt]9.3,4.5) rectangle ++(4pt,4pt);
\fill (10,8.5) circle (2pt); 
\fill (10.5,9) circle (2pt); 
\fill (11,9.5) circle (2pt); 
\fill (11.5,10) circle (2pt); 
\filldraw [deeppink] ([xshift=-3.5pt,yshift=-3.5pt]10.7,5.5) rectangle ++(7pt,7pt);
\filldraw ([xshift=-2pt,yshift=-2pt]10.7,5.5) rectangle ++(4pt,4pt);
\filldraw ([xshift=-2pt,yshift=-2pt]11.3,5.5) rectangle ++(4pt,4pt);
\filldraw [deeppink] ([xshift=-3.5pt,yshift=-3.5pt]11,4.8) rectangle ++(7pt,7pt);
\filldraw ([xshift=-2pt,yshift=-2pt]11,4.8) rectangle ++(4pt,4pt);
\fill (12,4.8) circle (2pt); 
\fill (10.05,3.9) circle (2pt); 
\fill (7.8,5) circle (2pt);
\draw (0,-0.3) node {\large $a$}; 
\draw (1,0.1) node {\large $b$}; 
\draw (2.5,1.2) node {\large $c$}; 
\draw (3.5,1.2) node {\large $d$}; 
\draw (4.5,1.2) node {\large $e$}; 
\draw (5.35,0.14) node {\large $f$}; 
\draw (6,0.2) node {\large $g$};  
\draw (7,-0.4) node {\large $h$}; 
\draw (7.5,2.2) node {\large $i$}; 
\draw (8.2,2.1) node {\large $j$}; 
\draw (8.6,0.4) node {\large $k$}; 
\draw (9.2,0.4) node {\large $l$}; 
\draw (9.8,6.2) node {\large $m$}; 
\draw (10.05,2.85) node {\large $n$}; 
\draw (12,3.65) node {\large $o$}; 
\draw (12.3,5.6) node {\large $p$}; 
\draw (1.75,4.1) node {\large  $r$}; 
\draw (6.2,5.4) node {\large  $s$}; 
\draw (7.28,3) node {\large  $u$}; 
\draw (7.52,5.2) node {\large  $t$}; 
\draw (8.28,4.28) node {\large  $v$}; 
\draw (9.3,2.5) node {\large  $w$}; 
\draw (9.5,4.8) node {\large  $x$}; 
\draw (10.36,3.7) node {\large  $y$}; 
\draw (10.9,4.45) node {\large  $z$}; 
\draw [thick] (0,0) -- (0,0.5); 
\draw [thick] (0,0.5) -- (0.5,2); 
\draw [thick] (1,0.5) -- (1,1); 
\draw [thick] (0,0.5) -- (1,1); 
\draw [thick] (1,1) -- (1,1.5); 
\draw [thick] (1,1.5) -- (0.5,2); 
\draw [thick] (1,1.5) -- (1.5,2); 
\draw [thick] (2.5,2) -- (2.5,1.5); 
\draw [thick] (3.5,2) -- (3.5,1.5); 
\draw [thick] (4.5,2) -- (4.5,1.5); 
\draw [thick] (0.5,2) -- (1,2.5); 
\draw [thick] (1,2.5) -- (2,4); 
\draw [thick] (5.8,2.8) -- (6.2,2.2); 
\draw [thick] (6.2,2.2) -- (5.4,0.5); 
\draw [thick] (6.2,2.2) -- (7.4,1.3); 
\draw [thick] (6.5,3.2) -- (6.5,1.5); 
\draw [thick] (6.5,1.5) -- (6,0.5); 
\draw [thick] (6.5,1.5) -- (7,0.5); 
\draw [thick] (7,0.5) -- (7,0); 
\draw [thick] (7,0.5) -- (7.4,1.3); 
\draw [thick] (2,4) -- (1.5,2); 
\draw [thick] (2,4) -- (2.5,2); 
\draw [thick] (2,4) -- (3.5,2); 
\draw [thick] (2,4) -- (4.5,2); 
\draw [thick] (2,4) -- (5.8,2.8); 
\draw [thick] (2,4) -- (6.5,3.2); 
\draw [thick] (7.4,1.3) -- (8.2,3); 
\draw [thick] (7.6,3) -- (7.6,2.5); 
\draw [thick] (8.2,3) -- (8.2,2.5); 
\draw [thick] (7.8,5) -- (7.6,3); 
\draw [thick] (7.8,5) -- (8.2,3); 
\draw [thick] (7.8,5) -- (6.5,3.2); 
\draw [thick] (7.8,5) -- (5.8,2.8); 
\draw [thick] (7.8,5) -- (4.5,2); 
\draw [thick] (7.8,5) -- (3.5,2); 
\draw [thick] (7.8,5) -- (2.5,2); 
\draw [thick] (7.8,5) -- (1.5,2); 
\draw [thick] (7.8,5) -- (1,2.5); 
\draw [thick] (6.25,5) -- (7.6,3); 
\draw [thick] (5,5.8) -- (6.25,5); 
\draw [thick] (7.2,5.8) -- (6.25,5); 
\draw [thick] (6.25,6.3) -- (5,5.8); 
\draw [thick] (5,5.8) -- (2,4); 
\draw [thick] (7.8,7) -- (6.25,6.3); 
\draw [thick] (7.8,7) -- (7.2,5.8); 
\draw [thick] (6.25,6.3) -- (7.2,5.8); 
\draw [thick] (7.8,5) -- (8.3,6); 
\draw [thick] (8.3,6) -- (8.8,7); 
\draw [thick] (8.8,7) -- (8.6,7.5); 
\draw [thick] (7.8,7) -- (8.6,7.5); 
\draw [thick] (9.5,8) -- (9.8,7); 
\draw [thick] (9.5,8) -- (8.6,7.5); 
\draw [thick] (9.8,7) -- (9.3,6.5); 
\draw [thick] (8.8,7) -- (9.3,6.5); 
\draw [thick] (9.3,6.5) -- (9.3,6); 
\draw [thick] (9.3,6) -- (8.9,5.5); 
\draw [thick] (8.9,5.5) -- (8.6,4.5); 
\draw [thick] (8.3,6) -- (8.6,4.5); 
\draw [thick] (8.6,4.5) -- (8.9,2.5); 
\draw [thick] (8.9,2.5) -- (8.9,1.5); 
\draw [thick] (8.9,1.5) -- (8.6,0.8); 
\draw [thick] (8.9,1.5) -- (9.2,0.8); 
\draw [thick] (9.3,4.5) -- (8.9,5.5); 
\draw [thick] (9.3,4.5) -- (8.9,2.5); 
\draw [thick] (9.5,8) -- (10,8.5); 
\draw [thick] (10,8.5) -- (10.5,9); 
\draw [thick] (10.5,9) -- (11,9.5); 
\draw [thick] (11,9.5) -- (11.5,10); 
\draw [thick] (11.5,10) -- (12.3,6); 
\draw [thick] (11,9.5) -- (12,4.8); 
\draw [thick] (10.5,9) -- (11.3,5.5); 
\draw [thick] (10,8.5) -- (10.7,5.5); 
\draw [thick] (9.3,6) -- (10.7,5.5); 
\draw [thick] (10.7,5.5) -- (11,4.8); 
\draw [thick] (11.3,5.5) -- (11,4.8); 
\draw [thick] (11.3,5.5) -- (12,4.8); 
\draw [thick] (12,4.8) -- (12,4); 
\draw [thick] (11,4.8) -- (10.05,3.9); 
\draw [thick] (9.3,4.5) -- (10.05,3.9); 
\draw [thick] (10.05,3.9) -- (10.05,3.2); 
\draw [thick] (9.8,7) -- (9.8,6.5); 
\fill [blue](0,0) circle (2pt); 
\fill [blue](1,0.5) circle (2pt); 
\fill [blue] (2.5,1.5) circle (2pt); 
\fill [blue](3.5,1.5) circle (2pt); 
\fill [blue](4.5,1.5) circle (2pt); 
\fill [blue](5.4,0.5) circle (2pt); 
\fill [blue](6,0.5) circle (2pt);  
\fill [blue](7,0) circle (2pt); 
\fill [blue](7.6,2.5) circle (2pt); 
\fill [blue](8.2,2.5) circle (2pt); 
\fill [blue](8.6,0.8) circle (2pt); 
\fill [blue](9.2,0.8) circle (2pt); 
\fill [blue](9.8,6.5) circle (2pt); 
\fill [blue](10.05,3.2) circle (2pt); 
\fill [blue](12,4) circle (2pt); 
\fill [blue](12.3,6) circle (2pt); 
\end{tikzpicture}}\hspace{.5cm} 
\subfigure[The bipartite graph associated to the network in (a).]{
\begin{tikzpicture}[scale=.8]
\filldraw ([xshift=-2pt,yshift=-2pt]0,2.7) rectangle ++(4pt,4pt);
\filldraw [deeppink] ([xshift=-3.5pt,yshift=-3.5pt]0,2.1) rectangle ++(7pt,7pt);
\filldraw ([xshift=-2pt,yshift=-2pt]0,2.1) rectangle ++(4pt,4pt);
\filldraw ([xshift=-2pt,yshift=-2pt]0,1.5) rectangle ++(4pt,4pt);
\filldraw [deeppink] ([xshift=-3.5pt,yshift=-3.5pt]0,0.9) rectangle ++(7pt,7pt);
\filldraw ([xshift=-2pt,yshift=-2pt]0,0.9) rectangle ++(4pt,4pt);
\filldraw [deeppink] ([xshift=-3.5pt,yshift=-3.5pt]0,0.3) rectangle ++(7pt,7pt);
\filldraw ([xshift=-2pt,yshift=-2pt]0,0.3) rectangle ++(4pt,4pt);
\filldraw [deeppink] ([xshift=-3.5pt,yshift=-3.5pt]0,-0.3) rectangle ++(7pt,7pt);
\filldraw ([xshift=-2pt,yshift=-2pt]0,-0.3) rectangle ++(4pt,4pt);
\filldraw [deeppink] ([xshift=-3.5pt,yshift=-3.5pt]0,-0.9) rectangle ++(7pt,7pt);
\filldraw ([xshift=-2pt,yshift=-2pt]0,-0.9) rectangle ++(4pt,4pt);
\filldraw ([xshift=-2pt,yshift=-2pt]0,-1.5) rectangle ++(4pt,4pt);
\filldraw [deeppink] ([xshift=-3.5pt,yshift=-3.5pt]0,-2.1) rectangle ++(7pt,7pt);
\filldraw ([xshift=-2pt,yshift=-2pt]0,-2.1) rectangle ++(4pt,4pt);
\filldraw [deeppink] ([xshift=-3.5pt,yshift=-3.5pt]0,-2.7) rectangle ++(7pt,7pt);
\filldraw ([xshift=-2pt,yshift=-2pt]0,-2.7) rectangle ++(4pt,4pt);

\fill  [color=deeppink] (4,4) circle (3.5pt); 
\fill [color=deeppink](4,3.6) circle (3.5pt); 
\fill[color=deeppink] (4,3.2) circle (3.5pt); 
\filldraw [deeppink] ([xshift=-3.5pt,yshift=-3.5pt]4,2.8) rectangle ++(7pt,7pt);
\filldraw ([xshift=-2pt,yshift=-2pt]4,2.8) rectangle ++(4pt,4pt);
\fill [color=deeppink](4,2.4) circle (3.5pt); 
\fill [color=deeppink](4,2) circle (3.5pt); 
\fill [color=deeppink](4,1.6) circle (3.5pt); 
\fill [color=deeppink](4,1.2) circle (3.5pt); 
\filldraw [deeppink] ([xshift=-3.5pt,yshift=-3.5pt]4,0.8) rectangle ++(7pt,7pt);
\filldraw ([xshift=-2pt,yshift=-2pt]4,0.8) rectangle ++(4pt,4pt);
\fill [color=deeppink](4,0.4) circle (3.5pt); 
\fill [color=deeppink](4,0.0) circle (3.5pt); 
\filldraw [deeppink] ([xshift=-3.5pt,yshift=-3.5pt]4,-0.4) rectangle ++(7pt,7pt);
\filldraw ([xshift=-2pt,yshift=-2pt]4,-0.4) rectangle ++(4pt,4pt);
\filldraw [deeppink] ([xshift=-3.5pt,yshift=-3.5pt]4,-0.8) rectangle ++(7pt,7pt);
\filldraw ([xshift=-2pt,yshift=-2pt]4,-0.8) rectangle ++(4pt,4pt);
\fill [color=deeppink](4,-1.2) circle (3.5pt); 
\filldraw [deeppink] ([xshift=-3.5pt,yshift=-3.5pt]4,-1.6) rectangle ++(7pt,7pt);
\filldraw ([xshift=-2pt,yshift=-2pt]4,-1.6) rectangle ++(4pt,4pt);
\fill [color=deeppink](4,-2) circle (3.5pt); 
\fill [color=deeppink](4,-2.4) circle (3.5pt); 
\fill [color=deeppink](4,-2.8) circle (3.5pt); 
\filldraw [deeppink] ([xshift=-3.5pt,yshift=-3.5pt]4,-3.2) rectangle ++(7pt,7pt);
\filldraw ([xshift=-2pt,yshift=-2pt]4,-3.2) rectangle ++(4pt,4pt);
\filldraw [deeppink] ([xshift=-3.5pt,yshift=-3.5pt]4,-3.6) rectangle ++(7pt,7pt);
\filldraw ([xshift=-2pt,yshift=-2pt]4,-3.6) rectangle ++(4pt,4pt);
\fill [color=deeppink](4,-4) circle (3.5pt); 

\fill (4,4) circle (2pt); 
\fill (4,3.6) circle (2pt); 
\fill (4,3.2) circle (2pt); 
\fill (4,2.4) circle (2pt); 
\fill (4,2) circle (2pt); 
\fill (4,1.6) circle (2pt); 
\fill (4,1.2) circle (2pt); 
\fill (4,0.4) circle (2pt); 
\fill (4,0.0) circle (2pt); 
\fill (4,-1.2) circle (2pt); 
\fill (4,-2) circle (2pt); 
\fill (4,-2.4) circle (2pt); 
\fill (4,-2.8) circle (2pt); 
\fill (4,-4) circle (2pt); 

\draw [thick] (0,2.7) -- (4,4); 
\draw [thick] (0,2.7) -- (4,3.6); 
\draw [thick] (0,2.1) -- (4,3.6); 
\draw [thick] (0,1.5) -- (4,3.2); 
\draw [thick] (0,1.5) -- (4,2.8); 
\draw [thick] (0,1.5) -- (4,2.4); 
\draw [thick] (0,1.5) -- (4,2); 
\draw [thick] (0,1.5) -- (4,1.6); 
\draw [thick] (0,1.5) -- (4,1.2); 
\draw [thick] (0,1.5) -- (4,0.4); 
\draw [thick] (0,0.9) -- (4,-1.2); 
\draw [thick] (0,0.3) -- (4,0.8); 
\draw [thick] (0,-0.3) -- (4,0); 
\draw [thick] (0,-0.9) -- (4,-2); 
\draw [thick] (0,-1.5) -- (4,-2); 
\draw [thick] (0,-1.5) -- (4,-2.8); 
\draw [thick] (0,-2.1) -- (4,-3.6); 
\draw [thick] (0,-2.7) -- (4,-2.8); 

\draw (-0.3,1.5) node {$r$}; 
\draw (-0.3,0.9) node {$s$}; 
\draw (-0.3,-0.9) node {$v$}; 
\draw (-0.3,-1.5) node {$x$}; 
\draw (-0.3,-2.7) node {$z$}; 

\draw (4.3,0.8) node {$s$}; 
\draw (4.3,-1.2) node {$u$}; 
\draw (4.33,-1.6) node {$v$}; 
\draw (4.3,-2) node {$w$}; 
\draw (4.3,-2.8) node {$y$}; 
\draw (4.3,-3.6) node {$z$}; 

\draw (0,4.5) node {$U$}; 
\draw (4,4.5) node {$R$}; 
\fill (4,-5) circle (0pt); 
\end{tikzpicture}}
\caption{\label{fig:violets} \rev{(a) A nonbinary phylogenetic network for all sixteen sections of the \emph{Viola} genus, based on the network published in~\cite{Violets2} and (b) the associated bipartite graph. The network in (a) contains a zig-zag path $(s,u,t)$ from which one can conclude by Corollary~\ref{cor:zigzag} that the network is not strictly tree-based. Moreover, it follows from Corollary~\ref{AantalkindNB} that the network is not even tree-based, because the three omnians~$v$, $x$ and~$z$ together only have two children. In the associated bipartite graph in (b), the vertices from~$U$ (the omnians) and the vertices from~$R$ (the reticulations) are drawn in the same order (from top to bottom) as they appear in the network (from left to right). Since there exists no matching that covers all vertices of~$U$, the \emph{Viola} network is not tree-based.}}
\end{figure}

First we note that the network is clearly not binary because the nodes labelled~$r$ and~$t$ are multifurcations. \rev{Next,} we show how Corollary~\ref{AantalkindNB} can be used to conclude that this network is not tree-based. Consider the three omnians labelled~$v$, $x$ and~$z$ (as before, omnians are indicated with square nodes).  Then the total number of different children of these three omnians is two:~$w$ and~$n$. Since the number of children is smaller than the number of considered omnians, it follows from Corollary~\ref{AantalkindNB} that the network is not tree-based.

\leo{Since the network is not tree-based, it can certainly not be strictly tree-based. To see this directly, we can apply Corollary~\ref{cor:zigzag}. The path $(s,u,t)$ is a zig-zag path starting at a reticulation that is also an omnian ($s$), zig-zagging via a reticulation ($u$), and ending at a multifurcation ($t$). The existence of such a path proves, by Corollary~\ref{cor:zigzag}, that the network is not strictly tree-based. }

\leo{To determine whether a network is tree-based or not, trying all subsets of the omnians is clearly not an efficient method. However, we can do this efficiently if we construct the bipartite graph associated to the network. For this example, the associated bipartite graph is displayed in Figure~\ref{fig:violets}(b). One can decide whether the network is tree-based by determining whether this bipartite graph has a matching that covers all vertices in~$U$, by Theorem~\ref{LstellingNB}. There exist simple polynomial-time algorithms for this task, see e.g.~\cite{schrijver}. In this case, no such matching exists and hence the network is not tree-based. Similarly, we can find out efficiently whether a network is strictly tree-based by constructing the modified bipartite graph and applying Theorem~\ref{thm:strictly}.}

\subsection{The origin of Eukaryotes}

\rev{The second example we discuss concerns the origin of Eukaryotes, which is displayed schematically in the phylogenetic network in Figure~\ref{fig:eukaryotes}. This figure is based on the network in~\cite{martin1999mosaic} and has been adapted to make it conform to the definition of phylogenetic networks used in this paper. The leaves of the network have been labelled arbitrarily by labels~$x_1,\ldots ,x_{41}$. Moreover, where the original network showed different prokaryotic genomes as differently colored lines inside the lineages, we show only the different lineages as black lines in Figure~\ref{fig:eukaryotes}. The six reticulations represent endosymbiosis; the merging of different prokaryotic genomes into a single lineage (present in the same cell). Horizontal gene transfer events between the lineages are not included in the network. The network is clearly nonbinary since the common ancestors of the Archaebacteria and the Eubacteria (the children of the root) are both multifurcations and, in addition, the child of the second reticulation from the top is also a multifurcation.}

\begin{figure}[h]
\center
\begin{tikzpicture}
\fill [color=deeppink](0,-3) circle (3.5pt); 
\fill [color=deeppink](1.5,-5) circle (3.5pt); 
\fill [color=deeppink](0.9,-8.2) circle (3.5pt); 
\fill [color=deeppink](-0.2,-7.8) circle (3.5pt); 
\fill [color=deeppink](0.9,-8.2) circle (3.5pt); 
\fill [color=deeppink](2,-9) circle (3.5pt); 
\fill [color=deeppink](-3.1,-9.5) circle (3.5pt);

\fill (0,1) circle (2pt);
\fill (-4,-1) circle (2pt);  
\fill (-4,-1.5) circle (2pt);
\fill (-5,-1.5) circle (2pt);
\fill (-2,-2) circle (2pt);
\fill (-2.4,-2.5) circle (2pt);
\fill (-4.4,-2) circle (2pt);
\fill (-3.6,-2) circle (2pt);
\fill (-1.5,-2.25) circle (2pt);
\fill (-1.9,-2.75) circle (2pt);
\fill (-1.4,-3) circle (2pt);
\fill (-1,-2.5) circle (2pt);
\fill (-0,-3) circle (2pt);   

\fill (4,-1) circle (2pt);    
\fill (4,-2) circle (2pt);
\fill (5,-1.5) circle (2pt);
\fill (2,-2) circle (2pt);
\fill (2.4,-2.5) circle (2pt);
\fill (5.4,-2) circle (2pt);
\fill (4.6,-2) circle (2pt);
\fill (1.5,-2.25) circle (2pt);
\fill (1.9,-2.75) circle (2pt);
\fill (1.4,-3) circle (2pt);
\fill (1,-2.5) circle (2pt);
\fill (4,-3) circle (2pt);
\fill (4.4,-3.5) circle (2pt);
\fill (3.6,-3.5) circle (2pt);
\fill (4.4,-2.5) circle (2pt);
\fill (4.9,-3) circle (2pt);

\fill (0,-3.5) circle (2pt); 
\fill (0.5,-4) circle (2pt);
\fill (1,-4.5) circle (2pt);
\fill (-1,-4) circle (2pt);
\fill (-1.4,-4.5) circle (2pt);
\fill (-0.6,-4.5) circle (2pt);
\fill (0.5,-5) circle (2pt);
\fill (0,-4.5) circle (2pt);
\fill (-0.75,-5.25) circle (2pt);
\fill (-1.15,-5.75) circle (2pt);
\fill (-0.35,-5.75) circle (2pt);
\fill (0.1,-5.5) circle (2pt);
\fill (0.9,-5.5) circle (2pt);

\fill (1.5,-5) circle (2pt); 

\fill (1.5,-6) circle (2pt);
\fill (4,-6.5) circle (2pt);
\fill (4.5,-7) circle (2pt);
\fill (4.9,-7.5) circle (2pt);
\fill (4.1,-7.5) circle (2pt);
\fill (2.2,-6.7) circle (2pt);
\fill (3,-7.3) circle (2pt);
\fill (3.4,-7.8) circle (2pt);
\fill (2.6,-7.8) circle (2pt);
\fill (0.5,-6.5) circle (2pt);
\fill (0.5,-7.8) circle (2pt);
\fill (3,-8.2) circle (2pt);
\fill (4,-9) circle (2pt);
\fill (4.5,-9.5) circle (2pt);
\fill (5.25,-10.25) circle (2pt);
\fill (5.75,-10.75) circle (2pt);
\fill (4.75,-10.75) circle (2pt);
\fill (4.35,-11.25) circle (2pt);
\fill (5.15,-11.25) circle (2pt);
\fill (-1,-7.3) circle (2pt);
\fill (-0.5,-7.6) circle (2pt);
\fill (-0.2,-7.8) circle (2pt);   
\fill (0.9,-8.2) circle (2pt);   
\fill (0.9,-8.5) circle (2pt);
\fill (1.3,-9) circle (2pt);   
\fill (.5,-9) circle (2pt);
\fill (2,-9.5) circle (2pt);
\fill (2.4,-10) circle (2pt);
\fill (1.6,-10) circle (2pt);
\fill (4.1,-10) circle (2pt);
\fill (4.5,-10.5) circle (2pt);
\fill (3.7,-10.5) circle (2pt);
\fill (2,-9) circle (2pt);
\fill (3.6,-9.5) circle (2pt);
\fill (0.2,-9.5) circle (2pt);
\fill (-0.6,-9.5) circle (2pt);
\fill (-2,-9) circle (2pt);
\fill (-1.5,-9.5) circle (2pt);
\fill (-1,-10) circle (2pt);
\fill (-3.1,-9.5) circle (2pt);
\fill (-3.1,-10) circle (2pt);
\fill (-3.5,-10.5) circle (2pt);
\fill (-2.7,-10.5) circle (2pt);
\fill (-2,-10) circle (2pt);
\fill (-1.5,-10.5) circle (2pt);
\fill (-0.5,-10.5) circle (2pt);
\fill (-.2,-9) circle (2pt);

\draw (-5,-1.75) node {$x_1$}; 
\draw (-4.4,-2.25) node {$x_2$}; 
\draw (-3.6,-2.25) node {$x_3$}; 
\draw (-2.4,-2.75) node {$x_4$}; 
\draw (-1.9,-3.0) node {$x_5$}; 
\draw (-1.4,-3.25) node {$x_6$}; 
\draw (-1.4,-4.75) node {$x_7$}; 
\draw (-0.6,-4.75) node {$x_8$}; 
\draw (-1.23,-5.95) node {$x_9$}; 
\draw (-0.35,-5.95) node {$x_{10}$};
\draw (0.1,-5.7) node {$x_{11}$};
\draw (0.9,-5.7) node {$x_{12}$};
\draw (-3.5,-10.7) node {$x_{13}$}; 
\draw (-2.7,-10.7) node {$x_{14}$}; 
\draw (-2,-10.2) node {$x_{15}$}; 
\draw (-1.5,-10.7) node {$x_{16}$}; 
\draw (-0.5,-10.7) node {$x_{17}$}; 
\draw (-0.6,-9.7) node {$x_{18}$}; 
\draw (.2,-9.7) node {$x_{19}$}; 
\draw (0.5,-9.2) node {$x_{20}$}; 
\draw (1.3,-9.2) node {$x_{21}$}; 
\draw (1.6,-10.2) node {$x_{22}$}; 
\draw (2.4,-10.2) node {$x_{23}$}; 
\draw (3.6,-9.7) node {$x_{24}$}; 
\draw (3.7,-10.7) node {$x_{25}$}; 
\draw (4.41,-10.7) node {$x_{26}$}; 
\draw (4.35,-11.45) node {$x_{27}$}; 
\draw (5.15,-11.45) node {$x_{28}$}; 
\draw (5.75,-10.95) node {$x_{29}$}; 
\draw (2.6,-8) node {$x_{30}$}; 
\draw (3.4,-8) node {$x_{31}$}; 
\draw (4.1,-7.7) node {$x_{32}$}; 
\draw (4.9,-7.7) node {$x_{33}$}; 
\draw (1.5,-3.25) node {$x_{34}$}; 
\draw (2,-3) node {$x_{35}$}; 
\draw (2.5,-2.75) node {$x_{36}$}; 
\draw (3.6,-3.7) node {$x_{37}$}; 
\draw (4.4,-3.7) node {$x_{38}$}; 
\draw (4.9,-3.2) node {$x_{39}$}; 
\draw (4.6,-2.2) node {$x_{40}$}; 
\draw (5.4,-2.2) node {$x_{41}$}; 

\draw [thick] (0,1) -- (-4,-1);
\draw [thick] (-4,-1) -- (-2,-2); 
\draw [thick] (-4,-1) -- (-5,-1.5);
\draw [thick] (-4,-1) -- (-4,-1.5);
\draw [thick] (-4,-1.5) -- (-4.4,-2);
\draw [thick] (-4,-1.5) -- (-3.6,-2);
\draw [thick] (-2,-2) -- (-2.4,-2.5);
\draw [thick] (-2,-2) -- (-1.5,-2.25);
\draw [thick] (-1.5,-2.25) -- (-1.9,-2.75);
\draw [thick] (-1.5,-2.25) -- (-1,-2.5);
\draw [thick] (-1,-2.5) -- (-1.4,-3);
\draw [thick] (-1,-2.5) -- (0,-3);

\draw [thick] (0,1) -- (4,-1);   
\draw [thick] (4,-1) -- (3,-1.5);
\draw [thick] (4,-1) -- (5,-1.5);
\draw [thick] (4,-1) -- (4,-2);
\draw [thick] (5,-1.5) -- (5.4,-2);
\draw [thick] (5,-1.5) -- (4.6,-2);

\draw [thick] (3,-1.5) -- (2,-2);

\draw [thick] (1,-2.5) -- (1.4,-3);
\draw [thick] (1,-2.5) -- (0,-3);
\draw [thick] (4,-2) -- (4.4,-2.5);
\draw [thick] (4.4,-2.5) -- (4.9,-3);
\draw [thick] (4.4,-2.5) -- (4,-3);
\draw [thick] (4,-3) -- (3.6,-3.5);
\draw [thick] (4,-3) -- (4.4,-3.5);

\draw [thick] (2,-2) -- (2.4,-2.5);
\draw [thick] (2,-2) -- (1.5,-2.25);
\draw [thick] (1.5,-2.25) -- (1.9,-2.75);
\draw [thick] (1.5,-2.25) -- (1,-2.5);
\draw [thick] (1,-2.5) -- (1.4,-3);
\draw [thick] (1,-2.5) -- (0,-3);
\draw [thick] (0,-3) -- (0,-3.5); 
\draw [thick] (0,-3.5) -- (-1,-4);
\draw [thick] (0,-3.5) -- (0.5,-4);
\draw [thick] (-1,-4) -- (-1.4,-4.5);
\draw [thick] (-0.6,-4.5) -- (-1,-4);
\draw [thick] (0.5,-4) -- (0,-4.5);
\draw [thick] (0.5,-4) -- (1,-4.5);
\draw [thick] (1,-4.5) -- (0.5,-5);
\draw [thick] (1,-4.5) -- (1.5,-5);
\draw [thick] (0.5,-5) -- (0.1,-5.5);
\draw [thick] (0.5,-5) -- (0.9,-5.5);
\draw [thick] (0,-4.5) -- (-0.75,-5.25);
\draw [thick] (-0.75,-5.25) -- (-0.35,-5.75);
\draw [thick] (-0.75,-5.25) -- (-1.15,-5.75);
\draw [thick] (4,-2) -- (1.5,-5);

\draw [thick] (1.5,-5) -- (1.5,-6); 
\draw [thick] (1.5,-6) -- (4,-6.5);
\draw [thick] (4,-6.5) -- (4.5,-7);
\draw [thick] (4.5,-7) -- (4.9,-7.5);
\draw [thick] (4.1,-7.5) -- (4.5,-7);
\draw [thick] (1.5,-6) -- (2.2,-6.7);
\draw [thick] (2.2,-6.7) -- (3,-7.3);
\draw [thick] (3,-7.3) -- (2.6,-7.8);
\draw [thick] (3,-7.3) -- (3.4,-7.8);
\draw [thick] (1.5,-6) -- (0.5,-6.5);
\draw [thick] (0.5,-6.5) -- (0.5,-7.8);
\draw [thick] (3,-8.2) -- (0.5,-7.8);
\draw [thick] (3,-8.2) -- (4,-9);
\draw [thick] (4,-9) -- (4.5,-9.5);
\draw [thick] (4.5,-9.5) -- (5.25,-10.25);
\draw [thick] (5.25,-10.25) -- (5.75,-10.75);
\draw [thick] (4.75,-10.75) -- (5.25,-10.25);
\draw [thick] (4.75,-10.75) -- (5.15,-11.25);
\draw [thick] (4.75,-10.75) -- (4.35,-11.25);
\draw [thick] (0.5,-6.5) -- (-1,-7.3);
\draw [thick] (-1,-7.3) -- (-0.5,-7.6);
\draw [thick] (-0.5,-7.6) -- (-0.2,-7.8);
\draw [thick] (-0.2,-7.8) -- (4,-6.5);
\draw [thick] (-0.2,-7.8) -- (-0.2,-9);
\draw [thick] (-0.2,-9) -- (0.2,-9.5);
\draw [thick] (-0.2,-9) -- (-0.6,-9.5);
\draw [thick] (0.5,-7.8) -- (0.9,-8.2);  
\draw [thick] (0.9,-8.2) -- (2.2,-6.7);
\draw [thick] (0.9,-8.2) -- (0.9,-8.5);
\draw [thick] (0.9,-8.5) -- (1.3,-9);
\draw [thick] (0.9,-8.5) -- (0.5,-9);
\draw [thick] (3,-8.2) -- (2,-9);
\draw [thick] (2,-9.5) -- (2,-9);
\draw [thick] (2,-9.5) -- (2.4,-10);
\draw [thick] (2,-9.5) -- (1.6,-10);
\draw [thick] (4.1,-10) -- (4.5,-9.5);
\draw [thick] (4.1,-10) -- (3.7,-10.5);
\draw [thick] (4.1,-10) -- (4.5,-10.50);
\draw [thick] (3.6,-9.5) -- (4,-9);

\draw [thick] (-2,-9) -- (-1.5,-9.5);
\draw [thick] (-1.5,-9.5) -- (-1,-10);
\draw [thick] (-1,-10) -- (-0.5,-10.5);
\draw [thick] (-1,-10) -- (-1.5,-10.5);
\draw [thick] (-1.5,-9.5) -- (-2,-10);
\draw [thick] (-3.1,-9.5) -- (-3.1,-10);
\draw [thick] (-3.1,-10) -- (-3.5,-10.5);
\draw [thick] (-3.1,-10) -- (-2.7,-10.5);
\draw [thick] (-0.5,-7.6) -- (-2,-9);
\draw [thick] (-2,-9) -- (-3.1,-9.5);
\draw [thick] (-1,-7.3) -- (-3.1,-9.5);
\draw [thick] (0,-4.5) -- (2,-9);

\node at (-4, 0)   (a) {Archaebacteria};
\node at (3.8, 0)   (b) {Eubacteria};
\node at (-.7, -6.4)   (c) {Eukaryotes};
\end{tikzpicture}
\caption{\label{fig:eukaryotes} \rev{A nonbinary phylogenetic network schematically illustrating the origin of Eukaryotes from Archaebacteria and Eubacteria, created by adapting the network published in~\cite{martin1999mosaic}. Taxa have been labelled $x_1,\ldots ,x_{41}$ arbitrarily. The six reticulations represent the merging of different prokaryotic genomes into a single lineage. Using Corollary~\ref{AantalkindNB}, it can be easily seen that the network is tree-based, since it has no omnians. Moreover, it follows from Corollary~\ref{cor:zigzag} that the network is even strictly tree-based.}}
\end{figure}
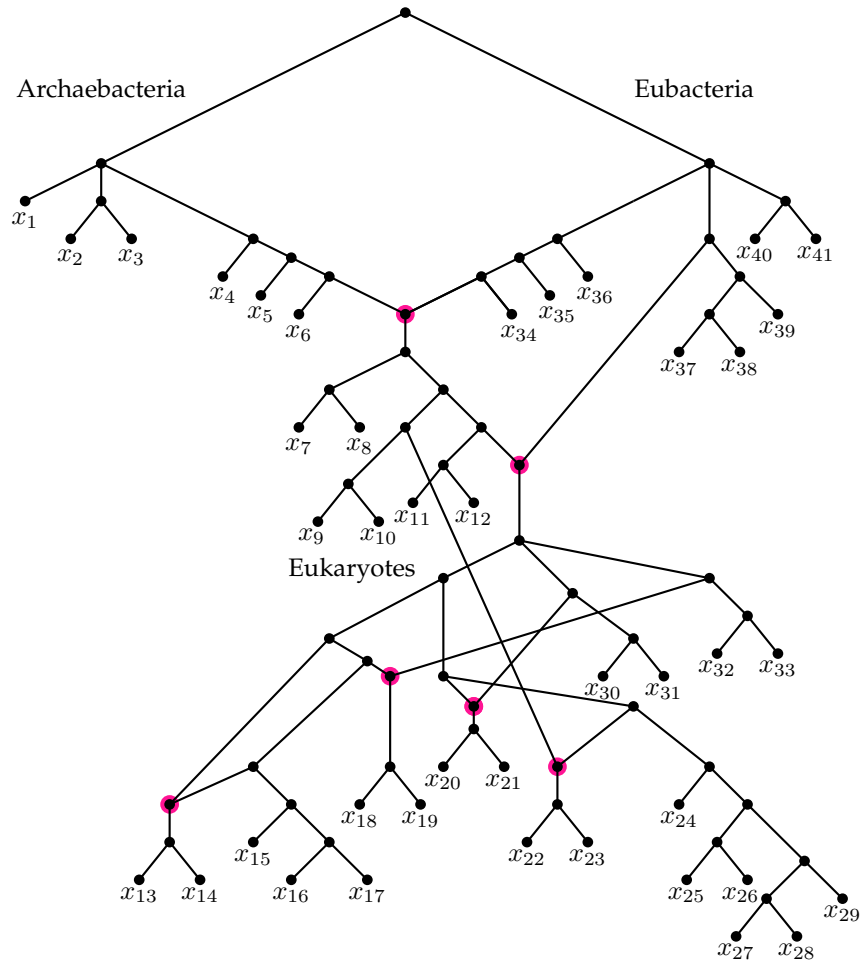

\rev{As usual, the reticulations are indicated in the figure by pink shading around the nodes. There are no omnians in this network since each non-leaf vertex has at least one non-reticulate child. Therefore, it follows directly from Corollary~\ref{AantalkindNB} that the network is tree-based. Moreover, we can use Corollary~\ref{cor:zigzag} to determine whether the network is strictly tree-based. Since all reticulations have indegree-2, there are no omnians, and there is no zig-zag path (multifucation -- reticulation -- multifurcation), we can conclude directly from Corollary~\ref{cor:zigzag} that the network is strictly tree-based. This means that the network can be seen as a base-tree augmented with linking-arcs (representing endosymbiosis events) between branches of the base-tree.}


\section{Discussion}\label{sec:discussion} 

As this is the first paper on tree-basedness for nonbinary phylogenetic networks, we end with a short discussion of our definitions and results. There are different ways to extend the concept of tree-based networks to the nonbinary case.

The most general variant allows linking arcs to be attached to vertices of the base-tree as well as to ``attachment points'' that subdivide the edges of the base-tree, and also allows several linking arcs to attach to the same vertex or attachment point. This can lead to vertices with more than two incoming arcs, and to vertices with more than two outgoing arcs, even if the base-tree is binary. Intuitively, this means that the non-binarity of the network can come both from the base-tree as well as from the way the linking-arcs are attached. Networks that can be formed this way we named \emph{tree-based}. 

A second possibility is to look at all binary refinements of a nonbinary network and to check if at least one of them is a tree-based binary network, using the definition of Francis and Steel. It turns out that this definition is equivalent to the previous one. Thus, a nonbinary network is tree-based precisely if it has at least one binary refinement that is tree-based.

A more restrictive variant allows the non-binarity of the network only to originate from the base-tree. In this case, linking arcs are only allowed to be attached to attachment points that subdivide the edges of the base-tree, and not to the original vertices of the base-tree. Moreover, no two linking arcs can be attached to the same attachment point. We named the networks that can be formed this way \emph{strictly tree-based}. This name is used to express that this definition is more restrictive than the previous ones. In particular, all strictly-tree-based networks are \emph{semi-binary}, meaning that reticulations have exactly two incoming arcs.

Of course, there are more possibilities. One could, for example, allow linking arcs to be attached only to attachment points, but still allow different linking arcs to attach to the same attachment point. However, we have not studied such variants as the definitions above seem the most natural ones.

We have given a complete characterization of tree-based nonbinary phylogenetic networks in terms of ``omnians'', i.e. non-leaf vertices of which all children are reticulations. Moreover, this has \rev{also led} to a new characterization for tree-based binary networks, which is, in our opinion, even simpler than the previous characterization in terms of ``zig-zag paths''~\cite{zhang}. Moreover, we have used our results to derive a characterization in terms of zig-zag paths, similar to the one in~\cite{zhang}, for  tree-based binary networks and for strictly-tree-based nonbinary networks. We have also shown that zig-zag paths can not be used in the same way to characterize tree-based nonbinary networks.

On the algorithmic side, we have shown that it can be decided in polynomial time whether a given nonbinary phylogenetic network is tree-based and whether it is strictly tree-based. We used a different approach from the one by Francis and Steel~\cite{Artikel2} and Zhang~\cite{zhang}, thus also obtaining a new way to decide if a binary network is tree-based. Moreover, we believe that our new approach for binary tree-based networks can be very useful when trying to solve some of the open problems mentioned by Francis and Steel~\cite{Artikel2}. \leo{In particular, is it possible to calculate how many base-trees a given (binary or nonbinary) tree-based network has? \rev{In~\cite{jetten}, it was shown how the method from this paper can be used to derive an upper bound on this number.} Another question by Francis and Steel was whether one can decide in polynomial time if a given binary phylogenetic network~$N$ is tree-based with \rev{a given tree~$T$ as base-tree}. However, this problem was very recently shown to be NP-hard~\cite{FixedTrees}.}

\rev{Finally, we have shown how our theorems can be applied to real phylogenetic networks by presenting two biological examples. We have shown that the first considered network, displaying the evolutionary history of the \emph{Viola} genus, is not tree-based. This means that we cannot see this evolutionary history as a tree-like process augmented with horizontal events. The numerous polyploidisations make this evolutionary history inherently network-like. The second network that we considered, showing the origin of Eukaryotes from Eubacterial and Archaebacterial genomes, turned out to be tree-based (and strictly tree-based). Hence, this evolutionary history can indeed be explained by a tree-like process augmented with horizontal events. However, note that this network is a high-level schematic depiction of the origin of Eukaryotes and the actual evolutionary history is much more complex, especially due to numerous gene transfer events. Moreover, the purpose of these examples is not to draw biological conclusions regarding these evolutionary histories, but to illustrate how our theorems can be applied to real phylogenetic networks. Since both these networks are nonbinary, previously-known theorems could not be applied to them. Moreover, the second example showed that, even though the network looks rather complex, the concept of omnians made it very easy to conclude that it is in fact tree-based.} 

\ifCLASSOPTIONcompsoc
  \section*{Acknowledgments}
\else
  \section*{Acknowledgment}
\fi

The authors would like to thank Mike Steel for interesting and useful discussions on the topic of this paper \leo{and the anonymous reviewers for their constructive comments.}

\ifCLASSOPTIONcaptionsoff
  \newpage
\fi



\bibliographystyle{IEEEtran}

\bibliography{Referencelist}

\end{document}